\documentclass[twocolumn,10pt]{IEEEtran}

\usepackage{setspace}
\usepackage{multicol}
\usepackage[final]{graphicx}
\usepackage{graphics}
\usepackage{amsmath,amsthm,amssymb,algorithmic, algorithm}
\usepackage{multicol}
\usepackage[margin=1in]{geometry}
\usepackage{subfig}
\usepackage{color}
\usepackage{setspace}
\usepackage{hyperref}
\usepackage{epstopdf}
\usepackage{soul}
\hypersetup{colorlinks=true}

\newcommand{\email}[1]{\texttt{\href{mailto:#1}{#1}}}

\newcommand{\yce}[1]{{\color{black}{#1}}}
\newcommand{\yx}[1]{{\color{black}{#1}}}
\newcommand{\ag}[1]{{\color{black}{#1}}}
\newcommand{\ye}[1]{{\color{black}{#1}}}
\newcommand{\yao}[1]{{\color{black}{#1}}}

\newcommand{\yaor}[1]{{\color{black}{#1}}}

\hypersetup{colorlinks=true}

\newcommand{\ben}{\begin{eqnarray}}

\newcommand{\een}{\end{eqnarray}}

\newcommand{\sign}{\mathop{\bf sgn}}
\newcommand{\transpose}{^{\top}}
\newtheorem{thm}{Theorem}
\newtheorem{corollary}{Corollary}
\newtheorem{lemma}{Lemma}

\newcommand{\vect}{\boldsymbol} 
\newcommand{\SNR}{\textsf{SNR}}

\newcommand{\sgn}{\textbf{sgn}}

\newcommand{\y}{\vect{y}}

\newcommand{\w}{\vect{w}}
\newcommand{\ab}{\vect{a}}
\newcommand{\bb}{\vect{b}}

\newcommand{\xb}{\vect{x}}

\newcommand{\zb}{\vect{z}}
\newcommand{\A}{\vect{A}}

\newcommand{\G}{\vect{G}}
\newcommand{\R}{\vect{R}}
\newcommand{\T}{\vect{T}}
\newcommand{\X}{\vect{X}}
\newcommand{\I}{\vect{I}}

\newcommand{\M}{\vect{M}}

\newcommand{\W}{\vect{W}}
\newcommand{\yb}{\vect{y}}

\newcommand{\vb}{\vect{v}}
\newcommand{\wb}{\vect{w}}
\newcommand{\la}{\langle}
\newcommand{\ra}{\rangle}

\title{Reduced-Dimension Multiuser Detection}

\author{Yao Xie,\thanks{Yao Xie (Email: \email{yao.c.xie@gmail.com}) was with the Department of Electrical Engineering at the Stanford University, and is currently with the Department of Electrical and Computer Engineering at the Duke University.}\quad
\and  Yonina C. Eldar,\thanks{Yonina C. Eldar  (Email: \email{yonina@ee.technion.ac.il}) is with the Department of Electrical Engineering, Technion, Israel Institute of Technology, and was visiting
the Department of Electrical Engineering at the Stanford University.} \and
\quad Andrea Goldsmith
\thanks{Andrea Goldsmith (Email: \email{andrea@wsl.stanford.edu}) is with the Department of
Electrical Engineering at the Stanford University.}
\thanks{This work is partially supported by the Interconnect Focus Center of the Semiconductor Research Corporation, BSF Transformative Science Grant 2010505, AFOSR grant FA9550-08-1-0010, and a Stanford General Yao-Wu Wang Graduate Fellowship. The paper was presented [in part] at the 48th Annual Allerton Conference on Communication, Control, and Computing in July 2010, and the  IEEE International Conference on communications (ICC) in June 2012.}
} \date{\today}

\date{\today}

\begin{document}

\singlespacing

\maketitle
\begin{abstract}

    We present a reduced-dimension multiuser detector (RD-MUD) structure \yao{for synchronous systems} that significantly decreases the number of required correlation branches at the receiver front-end, while still achieving performance similar to that of the conventional matched-filter (MF) bank. RD-MUD exploits the fact that, in some wireless systems, the number of active users may be small relative to the total number of users in the system. Hence, the ideas of analog compressed sensing may be used to reduce the number of correlators. The correlating signals used by each correlator are chosen as an appropriate linear combination of the users' spreading waveforms.  We derive the probability-of-symbol-error when using two methods for recovery of active users and their transmitted symbols: the reduced-dimension decorrelating (RDD) detector, which combines subspace projection and thresholding to determine active users and sign detection for data recovery, and the reduced-dimension decision-feedback (RDDF) detector, which combines decision-feedback  matching pursuit for active user detection and sign detection for data recovery. 
    We derive probability of error bounds for both detectors, and show that the number of correlators needed to achieve a small probability-of-symbol-error is on the order of the logarithm of the number of users in the system. The theoretical performance results are validated via numerical simulations.
    
\end{abstract}

\section{Introduction}\label{sec:intro}

Multiuser detection (MUD) is a classical problem in multiuser communications and signal processing (see, e.g., \cite{verduMUD1998,SchlegalGrant2006,Honig2009} and the references therein.) In multiuser systems, the users communicate simultaneously with a given receiver by modulating information symbols onto their unique signature waveforms. The received signal consists of a noisy version of the superposition of the transmitted waveforms. MUD has to detect the symbols of all users simultaneously.

Despite the large body of work on MUD,  it is not yet widely implemented in practice, largely due to its complexity and high-precision \yx{analog-to-digital} (A/D) requirements. The complexity arises both in the A/D as well as in the digital signal processing for data detection of each user. A conventional MUD structure consists of a matched-filter (MF) bank front-end followed by a linear or nonlinear digital multiuser detector. The MF-bank is a set of correlators, each correlating the received signal with the signature waveform of a different user. The number of correlators is therefore equal to the number of users. In a typical communication system, there may be thousands of users. We characterize the A/D complexity by the number of correlators at the receiver front-end, and measure data detection complexity by the number of real floating point operations required per \yce{decision bit} \cite{LupasVerdu1989} from the MF-bank output.

Verd\'{u}, in the landmark paper \cite{Verdu1986}, established the maximum likelihood sequence estimator (MLSE) as the MUD detector  minimizing the probability-of-symbol-error for data detection.  \yx{However, the complexity-per-bit of MLSE is exponential in the number of users when the signature waveforms are nonorthogonal.} To address the complexity issue, other \yx{low-complexity} suboptimal detectors have been developed, \yx{including the nonlinear}
decision feedback (DF) detector \cite{Varanasi1999} and \yx{linear detectors}.
The non-linear DF detector is based on the idea of interference cancellation, which decodes symbols iteratively by subtracting the detected symbols of strong users first to facilitate detection of weak users. The DF detector is a good compromise between complexity and performance (see, e.g., \cite{Varanasi1999}). We will therefore analyze the DF detector below as an example of a nonlinear \yx{digital} detector, but in a reduced dimensional setting. 

\yx{Linear detectors apply a linear transform to the receiver front-end output and then detect each user separately}. They have lower complexity than nonlinear methods but also worse performance. \ag{There are multiple} linear MUD techniques. 
The single-user detector is the simplest linear detector, \yx{however it suffers from user interference when signature waveforms are nonorthogonal.}
%
A linear detector that eliminates user inference is the decorrelating detector, 
which, for each user, projects the received signal onto the subspace associated with the signature waveform of that user. This projection amplifies noise when the signature waveforms are nonorthogonal. 
The decorrelating detector provides the best joint estimate of symbols and amplitudes in the absence of knowledge of the complete channel state information \cite{LupasVerdu1989}.
\yx{The \yce{Minimum Mean-Square-Error} (MMSE) detector takes into account both background noise and interference, and hence to some extent mitigates the noise amplification of the decorrelating detector in  low and medium SNR} \cite{Varanasi1999}. 
Because of \yx{the simplicity and interference elimination capability} of the decorrelating detector, we will focus on this detector as an example of a linear detector in the reduced-dimensional setting.

In many applications, the number of active users, $K$, can be much smaller than the total number of users, $N$ \cite{ApplebaumBajwaDuarte2011, FletcherRanganGoyal2010}. This analog signal sparsity allows the use of techniques from analog compressed sensing \cite{EldarSI2009, DuarteEldar2011} in order to reduce the number of correlators. While such sparsity has been exploited in various detection settings, there is still a gap in applying these ideas to the multiuser setting we consider here. Most existing work on exploiting compressed sensing \cite{CandesTao2006,Donoho2006} for signal detection assumes discrete signals and then applies compressed sensing via matrix multiplication
\cite{FletcherRanganGoyal2010,FletcherRanganGoyal2009,JinKimRao2010,HauptNowak2006,ZhuGiannakis2011}. 
In contrast, in multiuser detection the \ag{received} signal is continuous. Another complicating factor \ag{relative to previous work} is that here noise is added in the analog domain prior to \ag{A/D conversion at the front-end, which corresponds to the measurement stage in compressed sensing}. Therefore, A/D conversion will affect both signal and noise. Due to the MFs \ag{at the front-end}, the output noise vector is in general colored. Furthermore, it cannot be whitened without modifying the MF coefficient matrix, which corresponds to the measurement matrix in compressed sensing. In the discrete compressed sensing literature, it is usually assumed that white noise is added after measurement. \ye{An exception is the work of \cite{CastroEldar2011}.}  Finally, typically in compressed sensing the goal is to \textit{reconstruct} a sparse signal \ag{from its samples}, whereas in MUD the goal is to \textit{detect} both active users and their symbols. To meet the goal of MUD  we therefore adapt algorithms from compressed sensing for detection and develop results on the probability-of-symbol-error rather than on the mean-squared error (MSE).

In this work, we develop a low complexity MUD structure which we call  a reduced-dimension multiuser detector (RD-MUD) exploiting analog signal sparsity, assuming symbol-rate synchronization. The RD-MUD reduces the front-end \ag{receiver} complexity by decreasing the number of correlators without  increasing the complexity of digital signal processing, while still achieving performance similar to that of conventional MUDs that are based on the MF-bank front-end. The RD-MUD converts the \ag{received} analog signal into a discrete \ye{output} by correlating it with $M \ll N$ correlating signals. We incorporate analog compressed sensing techniques \cite{EldarSI2009} by forming the correlating signals as linear combinations of the signature waveforms via a (possibly complex) coefficient matrix $\A$.  
The RD-MUD output can thus be viewed as a projection of the MF-bank output onto a lower dimensional detection subspace. We then develop several digital detectors to detect both active users and their transmitted symbols, by combining ideas from compressed sensing and conventional MUD.   We study two such detectors in detail: the reduced-dimension decorrelating (RDD) detector, a linear detector that combines subspace projection and thresholding to determine active users with a sign detector for data recovery \cite{GribonvalMailheRauhut2007,BlumensathDavies2009}, and the reduced-dimension decision-feedback (RDDF) detector, a nonlinear detector that combines decision-feedback  matching pursuit (DF-MP) \cite{PatiRezaiifar1993,Tropp2004} for active user detection with sign detection for data recovery in an iterative manner. DF-MP differs from conventional MP \cite{PatiRezaiifar1993,Tropp2004} in that in each iteration the binary-valued detected symbols, rather than the real-valued estimates, are subtracted from the received signal to form the residual used by the next iteration. 

We provide \ag{probability-of-symbol-error} performance \ag{bounds} for these detection algorithms, using the coherence of the matrix $\A$ in a non-asymptotic regime with a fixed number of \ag{total} users and active users. Based on these results, we develop a lower bound on the number of correlators $M$ needed to attain a certain \ag{probability-of-symbol-error} performance. When $\A$ is a random partial discrete Fourier transform matrix, the $M$ required by these two specific detectors is on the order of $\log N$ as compared to $N$ correlators  required for conventional MUD. We validate these theoretical results via numerical examples.  Our analysis is closely related to \cite{Ben-HaimEldarElad2010}. However, \cite{Ben-HaimEldarElad2010} considers estimation in white noise, which differs from our problem in the aforementioned aspects. Our work also differs from prior results on compressed sensing for MUD, such as Applebaum et.al. \cite{ApplebaumBajwaDuarte2011} and Fletcher et.al. \cite{FletcherRanganGoyal2010,ApplebaumBajwaDaurte2010}, where a \ag{so-called} on-off random access channel is considered. In these references, the goal is to detect which users are active, and there is no need to detect the transmitted symbols as we consider here. Neither of these works consider front-end complexity.

In this paper we focus on a synchronous MUD channel model \cite{LupasVerdu1989}, where the transmission rate of all users is the same and their symbol epochs are perfectly aligned.  This user synchronization can be achieved using GPS as well as distributed or centralized synchronization schemes (see, e.g., \cite{BeekSync1999, Morelli2004}). Such methods are commonly used in cellular systems, ad-hoc networks, and sensor networks to achieve synchronization. 
Part of the MUD problem is signature sequence selection, for which there has also been a large body of work (see, e.g., \cite{VerduShamai1999}).  Here we do not consider optimizing signature waveforms so that our results are parameterized by the crosscorrelation properties of the signature waveforms used in our design.

The rest of the paper is organized as follows. Section \ref{sec:CD_principle} discusses the system model and reviews conventional detectors using the MF-bank front-end. Section \ref{sec:RD-MUD} presents the RD-MUD front-end and detectors. Section \ref{sec:RD-MUD-performance} contains the theoretical performance guarantee of two RD-MUD detectors: the RDD and RDDF detectors. Section \ref{sec:numerical_eg} validates the theoretical results through numerical examples.

\section{Background} \label{sec:CD_principle}

\subsection{Notation}

The notation we use is as follows. Vectors and matrices are denoted by boldface lower-case and boldface upper-case letters, respectively. 
The real and complex numbers are represented by $\mathbb{R}$ and $\mathbb{C}$, respectively.
The real part of a scalar $x$ is denoted as $\Re[x]$, and $x^*$ is its conjugate.
The set of indices of the nonzero entries of a vector $\xb$ is called the support of $\xb$. Given an index set $\mathcal{I}$, $\X_{\mathcal{I}}$ denotes the submatrix formed by the columns of a matrix $\X$ indexed by $\mathcal{I}$, and  $\xb_{\mathcal{I}}$ represents the subvector formed by the entries indexed by $\mathcal{I}$.
The identity matrix is denoted by $\vect{I}$. The transpose, conjugate transpose, and inverse of a matrix $\vect{X}$ are represented by $\vect{X}\transpose$, $\X^H$, and $\vect{X}^{-1}$, respectively, and $\X_{nm}$ denotes its $nm$ th value. 
%
%
The $\ell_2$ norm is denoted by $\|\vect{x}\| = (\vect{x}^H\vect{x})^{1/2}$, \yao{and the $\ell_\infty $ norm of a vector $\vect{x}$ is given by $\|\vect{x}\|_\infty = \max_{n=1}^N |x_n|$.}
The minimum and maximum eigenvalues of a positive-semidefinite matrix $\X$ are represented by $\lambda_{\min}(\X)$ and $\lambda_{\max}(\X)$, respectively. The trace of a square matrix $\X$ is denoted as $\mbox{tr}(\X)$. 
The notation $\mbox{diag}\{x_1, \ldots, x_n\}$ denotes a diagonal matrix with $x_1, \ldots, x_n$ on its diagonal.
We use $\I$ to denote the identity matrix and $\textbf{1}$ to denote an all-one vector.

The function $\delta_{nm}$ is defined such that $\delta_{nm} = 1$ only when $n = m$ and otherwise is equal to 0. 
The sign function is defined as $\sgn(x) = 1$, if $x > 0$, $\sgn(x) = -1$, if $x<0$, and otherwise is equal to 0.
The expectation of a random variable $x$ is denoted as $\mathbb{E}\{x\}$ and the probability of an event $\mathcal{A}$ is represented as $P(\mathcal{A})$. 
The union, intersection, and difference of two sets $\{A\}$ and $\{B\}$ are denoted by $\{A\}\cup \{B\}$, $\{A\} \cap \{B\}$, and $\{A\}/\{B\}$, respectively. 
The complement of a set $\{A\}$ is represented as $\{A\}^c$. 
The notation $A\subset B$ means that set $A$ is a subset of $B$. 
The \textit{inner product} (or \textit{crosscorrelation}) between two real analog signals $x(t)$ and $y(t)$ in $L_2$ is written as
$\la x(t), y(t)\ra = \frac{1}{T}\int_{0}^T x(t) y(t) dt,$ 
over the symbol time $T$. The $L_2$ norm of $x(t)$ is 
$\|x(t)\| = \langle x(t), x(t) \rangle^{1/2}$. 
Two signals are orthogonal if their crosscorrelation is zero.  
%

\subsection{System Model}\label{sec:model}

Consider a \ye{synchronous} multiuser system  \cite{verduMUD1998} with $N$ users. Each user is assigned a known unique signature waveform from a set $\mathcal{S} = \{s_n(\cdot): [0, T]\rightarrow \mathbb{R}, 1 \leq n \leq N\}$. 
Users modulate their data by their signature waveforms. 
There are $K$ active users out of $N$ possible users transmitting to the receiver. 
In our setting, we assume that active users modulate their signature waveforms using Binary Phase Shift Keying (BPSK)  modulation with the symbol of user $n$ denoted by $b_n \in \{1, -1\}$, for $n\in\mathcal{I}$, where $\mathcal{I}$ contains indices of all active users. The amplitude of the $n$th user's signal at the receiver is given by $r_n$, which is determined by the transmit power and the wireless channel gain. For simplicity, we assume $r_n$'s are real (but they can be negative), and known at the receiver.
The nonactive user can be viewed as transmitting with power $P_n = 0$, or equivalently transmitting zeros: $b_n = 0$, for $n\in\mathcal{I}^c$.
%
The received signal $y(t)$ is a superposition of the transmitted signals from the active users, plus white Gaussian noise $w(t)$ with zero-mean and variance $\sigma^2$:
\begin{equation}
    y(t) =  \sum_{n=1}^N r_n b_n s_n(t) + w(t), \qquad t \in [0, T], \label{sig_model}
\end{equation}
with $b_n \in \{1, -1\}$, $n\in\mathcal{I}$, and $b_n = 0$, $n \in \mathcal{I}^c$. The duration of the data symbol $T$ is referred to as the symbol time, which is also equal to the inverse of the data rate for binary modulation.

We assume that the signature waveforms are linearly independent. The crosscorrelations of the signature waveforms are characterized by the Gram matrix $\G$, defined as
\begin{equation}
[\G]_{n\ell}\triangleq \la s_n(t), s_\ell(t) \ra, \qquad 1\leq n, \ell \leq N.
\end{equation}
For convenience, we assume that $s_n(t)$ has unit energy: $\|s_n(t)\| = 1$ for all $n$ so that $[\G]_{nn} = 1$. Due to our assumption of linear independence of the signature waveforms, $\G$ is invertible. In addition, the signatures typically have low crosscorrelations, so that the magnitudes of the off-diagonal elements of $\G$ are much smaller than 1. 

Our goal is to detect the set of active users, i.e. users with indices in $\mathcal{I}$, and their transmitted symbols $\{b_n: n\in\mathcal{I}\}$. In practice the number of active users $K$ is typically much smaller than the total number of users $N$, which is a form of \yx{analog signal sparsity}.  As we will show, this sparsity enables us to reduce the number of correlators at the front-end and still be able to achieve performance similar to that of a conventional MUD using a bank of MFs. We will consider two scenarios: the case where $K$ is known, and the case where $K$ is bounded but unknown. The problem of estimating $K$ can be treated using techniques such as those in \cite{BiglieriLops2007}.

%

\begin{figure}[h]
    \centering
    \subfloat[]
    {\label{fig:MF_bank}
        \includegraphics[width=.55\linewidth]{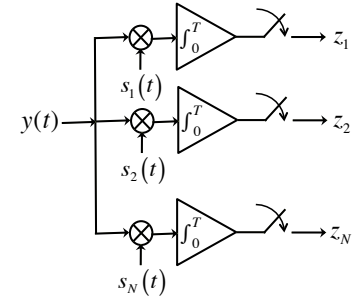}}
        \qquad
     \subfloat[]
     {\label{fig:RD_MUD}
          \includegraphics[width=.55\linewidth]{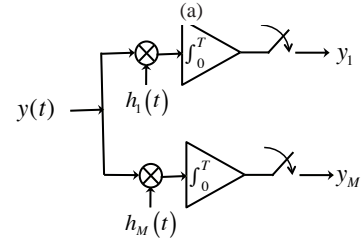}}
          \vspace{0.1in}
        \caption{Front-end of (a) conventional MUD using MF-bank, and (b) RD-MUD.}
        \label{Fig:front_end}
    \end{figure}
    
    \begin{figure}[h]
    \centering
    \subfloat[]
    {\label{fig:linear_detector}
        \includegraphics[width=.85\linewidth]{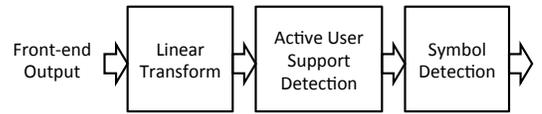}}
        \qquad
     \subfloat[]
     {\label{fig:non_linear_detector}
          \includegraphics[width=.55\linewidth]{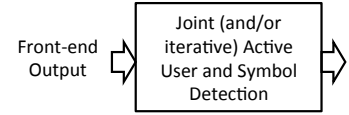}}
          \vspace{0.1in}
        \caption{The diagram of (a) linear detector, and (b) nonlinear detector.}
        
        \label{Fig:detector}
    \end{figure}

\subsection{Conventional MUD}

A conventional MUD detector has a MF-bank front-end followed by a digital detector. We now review this architecture.

\subsubsection{MF-bank front-end}

For general single-user systems, the MF
multiplies the received signal $y(t)$ with \yx{the single user waveform} $s(t)$ and integrates over a symbol time. 
The MF-bank is an extension of the MF for multiple users, and has $N$ MFs in parallel: the $n$th branch correlates the received signal with 
the corresponding signature waveform $s_n(t)$, as illustrated in Fig. \ref{fig:MF_bank}. 
%
The output of the MF-bank is a set of sufficient statistics for MUD when the amplitudes $r_n$ are known \cite{verduMUD1998}. 
%
Using (\ref{sig_model}), the output \yx{of the MF-bank can be written as}
\begin{equation}
    \zb = \G \R \bb + \vect{u},\label{MF_bank}
\end{equation}
\yx{where $\zb = [z_1, \cdots, z_N]\transpose$, $z_n = \la y(t), s_n(t) \ra$,  $\R\in\mathbb{R}^{N\times N}$ is a diagonal matrix with $[\R]_{nn} = r_n$, $\bb = [b_1, \cdots, b_N]\transpose$ and $\vect{u} = [u_1, \cdots, u_N]\transpose$, where $u_n = \la w(t), s_n(t) \ra$. The vector $\vect{u}$ is Gaussian distributed with zero mean and covariance  $\mathbb{E}\{\vect{u}\vect{u}^H\} = \sigma^2 \G$ (for derivation see \cite{verduMUD1998})}.

\subsubsection{MF-bank detection}\label{sec:MF_MUD}

Conventional MUD detectors based on the MF-bank output can be classified into two categories: linear and nonlinear, as  illustrated in Fig. \ref{Fig:detector}. In the literature, the synchronous MUD model typically assumes all users are active, i.e. $b_n \in \{1, -1\}$, and hence the goal of the MUD detectors is to detect all user symbols. 
The linear detector applies a linear transform $\vect{T}$ to the MF-bank output (illustrated in Fig. \ref{fig:linear_detector}):
\begin{equation}
\vect{T}\zb = \vect{T}\G \R \bb + \vect{T}\vect{u},  \label{MUD_output}
\end{equation}
and detects each user's symbol separately using a sign detector:
\begin{equation}
\hat{b}_n = \sign(r_n[\T\zb]_n), \qquad 1\leq n\leq N. \label{sign_detector}\end{equation}

Several commonly used linear detectors are the single-user detector, the decorrelating detector and the minimum-mean-square-error (MMSE) detector. The single-user detector \cite{LupasVerdu1989} is equivalent to choosing $\T = \I$ in (\ref{MUD_output}) and (\ref{sign_detector}). \yx{By applying a linear transform $\T = \G^{-1}$ in (\ref{MUD_output}), the decorrelating detector can remove the user interference and recover symbols perfectly in the absence of noise;} however, it also amplifies noise when $\G\neq \I$. 
The MMSE detector minimizes the MSE between the linear transform of the MF-bank output and symbols, and corresponds to $\T = (\G + \sigma^2\R^{-2})^{-1}$ in (\ref{MUD_output}) \cite{verduMUD1998}. 


Nonlinear detectors, on the other hand, detect symbols jointly and (or) iteratively as illustrated in Fig. \ref{fig:non_linear_detector}. Examples include MLSE and the successive interference cancellation (SIC) detector \cite{verduMUD1998}.
The MLSE solves the following optimization problem:
\begin{equation}
\max_{b_n \in \{1, -1\}} 2\yb^H\R\bb - \bb^H \R\G\R\bb. \label{ML_MF}
\end{equation}
However, when the signature waveforms are nonorthogonal this optimization problem is exponentially complex in the number of users \cite{verduMUD1998}. 
The SIC detector first finds the active user with the largest amplitude, detects its symbol, subtracts its effect from the received signal, and iterates the above process using the residual signal. After $N$ iterations, the SIC detector determines all users. 

\section{Reduced-Dimension MUD (RD-MUD)}\label{sec:RD-MUD}

The RD-MUD front-end, illustrated in Fig. \ref{fig:RD_MUD}, correlates the received signal $y(t)$ with a set of correlating signals $h_m(t)$, $m = 1, \cdots M$, where $M$ is typically much smaller than $N$.  The front-end output is processed by either a linear or nonlinear detector to detect active users and their symbols; \ye{the design of these detectors is adapted to take the analog sparsity into account.}

\subsection{RD-MUD front-end}\label{sec:RD-FE}
Design of the correlating signals $h_m(t)$ is key for RD-MUD to reduce the number of correlators. To construct these signals, we rely on the ideas introduced in \cite{EldarSI2009} to construct multichannel filters that sample \ye{structured} analog signals at sub-Nyquist rates. 
Specifically, we use the biorthogonal signals with respect to $\{s_n(t)\}$, which are defined as:
\begin{equation}
\hat{s}_n(t) = \sum_{\ell=1}^N  [\G^{-1}]_{n\ell} s_\ell(t), \qquad 1\leq n \leq N.
\end{equation}
These signals have the property that $\la s_n(t), \hat{s}_m(t) \ra = \delta_{nm}$, for all $n$, $m$. Note that when $\{s_n(t)\}$ are orthogonal, $\G =\I$ and $\hat{s}_n(t) = s_n(t)$.

The RD-MUD front-end uses as its correlating signals the functions
\begin{equation}
h_m(t) = \sum_{n = 1}^{N} a_{mn} \hat{s}_n(t), \qquad 1\leq m\leq M, \label{h_def}
\end{equation}
where $a_{mn}$ are (possibly complex) weighting coefficients. 
Define a coefficient matrix $\A \in \mathbb{R}^{M\times N}$ with $[\A]_{mn} \triangleq a_{mn}$ and denote the $n$th column of $\A$ as $\ab_n \triangleq [a_{1n}, \cdots, a_{Mn}]\transpose$, $n = 1, \cdots, N$. We normalize the columns of $\A$ so that $\|\ab_n\| = 1$.
The design of the correlating signals is equivalent to the design of  $\A$ for a given $\{s_n(t)\}$. Evidently, the performance of RD-MUD will depend on $\A$. We will use coherence as a measure of the quality of $\A$, which is defined as:
\begin{equation}
\mu\triangleq \max_{n\neq \ell}\left|\ab_n^H \ab_\ell \right|.\label{def_coherence}
\end{equation}
As we will show later in Section \ref{sec:single_user_det}, it is desirable that 
$\mu$ is small \ye{to guarantee small probability-of-symbol-error}. This requirement also reflects a tradeoff in choosing how many correlators to use in the RD-MUD front-end. With more correlators, the coherence of $\A$ can be lower and the performance of RD-MUD improves. 

Choosing the correlating signals (\ref{h_def}) and using the receive signal model (\ref{sig_model}), the output of the $m$th correlator is given by:
\begin{equation}
\begin{split}
       & y_{m} = \la h_m(t), y(t) \ra\\
&        = \left\la \sum_{n=1}^N a_{mn}
        \hat{s}_n(t), \sum_{\ell=1}^N r_\ell b_\ell s_\ell(t) \right\ra \\
        &+ \left\la \sum_{n = 1}^{N} a_{mn} \hat{s}_n(t), w(t) \right\ra \\
        &= \sum_{\ell =1}^N  r_\ell b_\ell \sum_{n = 1}^N a_{mn} \la \hat{s}_n(t), s_\ell(t)\ra + {w}_m \\
&        = \sum_{\ell=1}^N a_{m\ell}  r_\ell b_\ell + {w}_m, \end{split}
        \label{eq1}
\end{equation}
where \yx{the output noise is given by} $w_m \triangleq  \sum_{n = 1}^{N} a_{mn} \left\la\hat{s}_n(t), w(t) \right\ra$. 
Denoting $\yb = [y_{1}, \cdots, y_{M}]\transpose$ and ${\wb} = [w_1, \cdots,w_M]\transpose$, we can express the RD-MUD output (\ref{eq1}) in vector form as
\begin{equation}
    \yb = \A \R\bb + {\wb},\label{RD_MUD_model}
\end{equation}
where $\wb$ is a Gaussian random vector with zero mean and covariance $\sigma^2 \A\G^{-1}\A^H$ (for derivation see \cite{Xie2011PhD,XieEldarGoldsmith2010}). The vector $\yb$ can be viewed as a linear projection of the MF-bank front-end output onto a lower dimensional subspace which we call the \textit{detection subspace}. Since there are at most $K$ active users, $\bb$ has at most $K$ non-zero entries. The idea of RD-MUD is that when the original signal vector $\bb$ is sparse, with proper choice of the matrix $\A$, the detection performance for $\bb$ based on $\yb$ of (\ref{RD_MUD_model}) in the detection subspace can be similar to the performance based on the output of the MF-bank front-end $\zb$ of (\ref{MF_bank}).

\subsection{RD-MUD \yx{detection}}

We now discuss how to recover $\bb$ from the RD-MUD front-end output  using digital detectors. The model (\ref{RD_MUD_model}) for RD-MUD has a similar form to the observation model in the compressed sensing literature \cite{FletcherRanganGoyal2009,Ben-HaimEldarElad2010}, except that the noise in the RD-MUD front-end output is colored. Hence, to recover $\bb$, we can \yx{combine ideas developed in the context of compressed sensing and MUD}. 
The linear detector for RD-MUD first estimates active users $\hat{\mathcal{I}}$ using support recovery techniques from compressed sensing. 
%
Once the active users are estimated, 
we can write the RD-MUD front-end output model (\ref{RD_MUD_model}) as
\begin{equation}
\yb = \A_{\hat{\mathcal{I}}}\R_{\hat{\mathcal{I}}}\bb_{\hat{\mathcal{I}}} + \wb,\label{restrict}
\end{equation}
\yx{from which we can detect} the symbols $\bb_{\hat{\mathcal{I}}}$ by applying a linear transform to $\yb$. 
The nonlinear detector for RD-MUD detects active users and their symbols jointly (and/or iteratively).

We will focus on recovery based on two algorithms: (1) the RDD detector, a linear detector that uses subspace projection along with thresholding \cite{BlumensathDavies2009,FletcherRanganGoyal2009} to determine active users and sign detection for data recovery; (2) the RDDF detector, a nonlinear detector that combines decision-feedback matching pursuit (DF-MP) for active user detection and sign detection for data recovery. These two algorithms are summarized in \yx{Algorithms \ref{alg-1} and \ref{alg-2}}.

\subsubsection{RDD detector}\label{sec:RDD}

A natural strategy for detection is to compute the inner product $\ab^H\yb$ 
and  
\yx{detect active users by choosing indices corresponding to the $K$ largest magnitudes of these inner products:}
\begin{equation}
\begin{split}
&\hat{\mathcal{I}} = \{n: \quad \mbox{if $|\Re[\ab_n^H \yb]|$ is among the }\\
&\mbox{$K$ largest of $|\Re[\ab_n^H\yb]|$, $n = 1, \cdots, N$} \}.
\end{split}\label{support}
\end{equation}
\yx{This corresponds to the thresholding support recovery algorithm in compressed sensing (e.g. \cite{FletcherRanganGoyal2009}).}
To detect symbols, we use sign detection:
\begin{equation}
    \hat{b}_n = \left\{
    \begin{array}{cc}
     \sgn\left(r_n\Re[\ab_n^H \yb] \right), & n\in\hat{\mathcal{I}};\\
     0, &n\notin\hat{\mathcal{I}}.
     \end{array}\right.
    \label{RD-MUD-sign}
\end{equation}
In detecting active users (\ref{support}) and their symbols (\ref{RD-MUD-sign}), we take the real parts of the inner products because the imaginary part of $\ab_n^H \yb$ contains only noise and interference, since we assume that symbols $b_n$ and amplitudes $r_n$ are real and only $\A$ can be complex.  When $K = N$ and $M = N$, the RDD detector becomes the decorrelator in  conventional MUD.

To compute the complexity-per-bit of the RDD detector \ye{we note that} computing $\A^H \yb$ requires $MN$ floating point operations when $\A$ is real (or $2MN$ operations when $\A$ is complex) for detection of $N\log_2 3$ bits (since equivalently we are detecting $b_n\in \{-1, 0, 1\})$. Hence the complexity-per-bit of RDD is proportional to $M$. 
Since $M\leq N$ in RD-MUD, the complexity-per-bit of RDD (and other RD-MUD linear detectors as well) is \ye{lower than} that of the conventional \ye{decorrelating} linear MUD detector. Furthermore, RDD and other linear RD-MUD detectors require much lower complexity in the analog front-end. 

When the number of users is not known, we can replace Step 2 in Algorithm \ref{alg-1} by 
\begin{equation}
\hat{\mathcal{I}} = \{n \in \{1, \ldots, N\}: |\Re[\ab_n^H \yb]| > \xi \},
\end{equation}
where $\xi > 0$ is a chosen threshold.  We refer to this method as the RDD threshold (RDDt) detector. The RDDt detector is related to the OST algorithm for model selection in \cite{BajwaCalderbankJafarpour2010}. \yaor{The choice of $\xi$ depends on $r_n$, $\sigma^2$, $M$, $N$, $\mu$ and the maximum eigenvalue of $\G^{-1}$. Bounds on $\xi$ associated with error probability bounds will be given in Theorem \ref{thm_noisy}. In Section \ref{sec:numerical_eg} we explore numerical optimization of $\xi$, where we find that to achieve good performance, $\xi$ should increase with $N$ or $K$, and decrease with $M$.}

\algsetup{indent=2em}
 \begin{algorithm}[h]
 \caption{RDD detector}
 \begin{algorithmic}[1]
\STATE  \textbf{Input}: An $M\times N$ matrix $\A$, a vector $\yb \in \mathbb{C}^{M}$ and the number of active users $K$.
\STATE   Detect active users: find $\hat{\mathcal{I}}$ that contains indices of the $K$ largest  $|\Re[\ab_n^H \yb]|$.
\STATE  Detect symbols: $\hat{b}_n = \sign(r_n \Re[\ab_n^H \yb])$ for $n\in\hat{\mathcal{I}}$, and $\hat{b}_n = 0$ for $n\notin \hat{\mathcal{I}}$.
 \end{algorithmic}
 \label{alg-1}
 \end{algorithm}
\algsetup{indent=5em}
 \begin{algorithm}[h!]
 \caption{RDDF detector}
 \begin{algorithmic}[1]
 \STATE  \textbf{Input}: An $M\times N$ matrix $\A$, a vector $\yb \in \mathbb{C}^{M}$ and number of active users $K$.
\STATE Initialize: $\mathcal{I}^{(0)}$ is empty, $\bb^{(0)} = 0$, $\vb^{(0)} = \yb$.
\STATE Iterate Steps 4 -- 6 for $k = 1, \cdots, K$:
\STATE Detect active user: $n_k = \arg\max_{n}|\Re[\ab_n^H\vb^{(k-1)}]|$.
\STATE Detect symbol: ${b}_n^{(k)} = \sign(r_{n_k}\Re[\ab_{n_k}^H \vb^{(k-1)}])$, for $n = n_k$, and ${b}_n^{(k)} = {b}_n^{(k-1)}$ for $n \neq n_k$.
\STATE  Update: ${\mathcal{I}}^{(k)} = {\mathcal{I}}^{(k-1)}\cup\{n_k\}$,  \\and $\vb^{(k)} = \yb - \A\R\bb^{(k)}$.
\STATE Output: $\hat{\mathcal{I}} = \mathcal{I}^{(K)}$, $\hat{\bb} = \bb^{(K)}$.
 \end{algorithmic}
 \label{alg-2}
 \end{algorithm}

\subsubsection{RDDF detector}\label{sec:algorithm_II}

The RDDF detector determines active users and their corresponding symbols iteratively. It starts with an empty set as the initial estimate for the set of active user $\mathcal{I}^{(0)}$, zeros as the estimated symbol vector $\bb^{(0)} = \vect{0}$, and the front-end output as the residual vector $\vb^{(0)}=\y$. Subsequently, in each iteration $k = 1, \cdots, K$, the algorithm selects the column $\ab_n$ that is most highly correlated with the residual $\vb^{(k-1)}$ as the detected active user in the $k$th iteration: 
\begin{equation}
n_k=\arg\max_n \left|\Re[\ab_n^H \vb^{(k-1)}]\right|,\label{support_OMP}
\end{equation} 
which is then added to the active user set ${\mathcal{I}}^{(k)} = {\mathcal{I}}^{(k-1)}\cup\{n_k\}$.
The symbol for user $n_k$ is detected with other detected symbols staying the same:
\begin{equation}
{b}_{n}^{(k)} = \left\{
\begin{array}{cc} \sign\left(r_{n_k} \Re[\ab_{n_k}^H \vb^{(k-1)}]\right), & n = n_k;\\
{b}_n^{(k-1)}, & n \neq n_k.
\end{array}\right.\label{symbol_OMP}
\end{equation}
The residual vector is then updated through $\vb^{(k)} = \yb - \A\R\bb^{(k)}. $ 
The residual vector represents the part of $\bb$ that has yet to be detected by the algorithm along with noise. The iteration repeats $K$ times (as we will show, with high probability DF-MP never detects the same active user twice), and finally the active user set is given by $\hat{\mathcal{I}}={\mathcal{I}}^{(K)}$ with the symbol vector $\hat{b}_n = b^{(K)}_n$, $n = 1, \cdots, N$. When $K = N$ and $M = N$, the RDDF detector becomes the successive interference cancelation technique in the conventional MUD.

The complexity-per-bit of RDDF is proportional to $KM$. Since $M\leq N$, this implies that the complexity for data detection of RDDF is \ye{lower than} that of the conventional DF detector (the complexity-per-bit of the DF detector is proportional to $KN$). \ye{Note that RDDF is similar to MP in compressed sensing but with symbol detection.}

We can modify the RDDF detector to account for an unknown number of users by iterating only when the residual does not contain any significant ``component'', i.e., when $\|A^H \vb^{(k)}\|_\infty < \epsilon$ for some threshold $\epsilon > 0$. We refer to this method as the RDDF threshold (RDDFt) detector. \yaor{The choice of   $\epsilon$ depends on $r_n$, $\sigma^2$, $M$, $N$, $\mu$ and the maximum eigenvalue of $\G^{-1}$. As with the threshold $\xi$, bounds on $\epsilon$ to ensure given error probability bounds, and its numerical optimization, are presented in Theorem \ref{thm_noisy} and Section \ref{sec:numerical_eg}, respectively. As in the RDDt, here we also found in numerical optimizations that $\epsilon$ should increase with $N$ or $K$, and decrease with $M$.}

\subsubsection{Noise whitening transform}\label{sec:prew}

The noise in the RD-MUD output (\ref{RD_MUD_model}) is in general colored due to the matched filtering at the front-end. We can whiten the noise by applying a linear transform $(\A\G^{-1}\A^H)^{-1/2}$ before detecting active users and symbols, as illustrated in Fig. \ref{Fig:prewhitening}. 
The whitened output is given by:
\begin{equation}
\begin{split}
&\yb_w \triangleq  (\A\G^{-1}\A^H)^{-1/2}\yb \\
&=  (\A\G^{-1}\A^H)^{-1/2}\A\R\bb + \wb_0, \end{split}
\label{29}
\end{equation}
where $\wb_0$ is a Gaussian random vector with zero mean and covariance matrix $\sigma^2 \I$. If we define a new measurement matrix 
\begin{equation}
\A_w \triangleq (\A\G^{-1}\A^H)^{-1/2}\A,
\end{equation}
then the RDD and RDDF detectors \ye{can be applied by replacing} $\A$ with $\A_w$ and $\y$ with $\y_w$ in (\ref{support}), (\ref{RD-MUD-sign}), (\ref{support_OMP}) and (\ref{symbol_OMP}). While whitening the noise, the whitening transform also distorts the signal component. \yx{As we demonstrate via numerical examples in Section \ref{eg:prew}, noise whitening is beneficial when the signature waveforms highly correlated.}

    \begin{figure}[h]
    \centering
        \includegraphics[width=.8\linewidth]{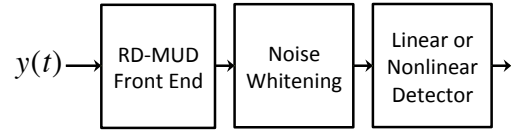}
        \caption{A MUD detector with noise whitening transform.}
        \label{Fig:prewhitening}
    \end{figure}

\subsubsection{Other RD-MUD linear detectors}\label{sec:RD_MUD_decorrelator}


\yx{By combining ideas developed in the context of compressed sensing and conventional linear MUD detection, we can develop alternative linear detectors in the reduced-dimension setting.}

\textit{Reduced-dimension MMSE (RD-MMSE) detector:} 
Similar to the MMSE detector of the conventional MUD, a linear detector based on the MMSE criterion can be derived for (\ref{restrict}). \yx{The RD-MMSE detector determines active users through the support recovery method of (\ref{support}), and then uses a linear transform $\M$ that minimizes $\mathbb{E}\{\|\bb_{\hat{\mathcal{I}}}-\M\y\|^2\}$ to estimate the symbols}. Here the expectation is with respect to the vector of transmitted symbols $\bb_{\hat{\mathcal{I}}}$ and the noise vector $\wb$. Following the approach for deriving the conventional MMSE detector \cite{verduMUD1998}, assuming that $\bb_{\hat{\mathcal{I}}}$ is uncorrelated with $\wb$ and $\mathbb{E}\{\bb_{\hat{\mathcal{I}}}\bb_{\hat{\mathcal{I}}}^H\} = \I$,
we obtain the linear transform for the reduced-dimension MMSE (RD-MMSE) detector as \yx{(see Appendix \ref{app:RD_MMSE} for details)}:
$\M = \R_{\hat{\mathcal{I}}}\A_{\hat{\mathcal{I}}}^H(\A_{\hat{\mathcal{I}}}\R_{\hat{\mathcal{I}}}^2\A_{\hat{\mathcal{I}}}^H + \sigma^2\A\G^{-1}\A^H)^{-1}. $
The symbols are then determined as:
\begin{equation}
\hat{b}_n =
\left\{\begin{array}{cc}
\sign(r_n\Re\{[\yx{\M}\y]_n\}), & n\in\hat{\mathcal{I}};\\
0, & n\notin\hat{\mathcal{I}}.
\end{array}\right. \label{MMSE_RDMUD}
\end{equation}
Similarly, we can modify RDDF by replacing symbol detection by (\ref{MMSE_RDMUD}) on the detected support $\mathcal{I}^{(k)}$ in each iteration.

\textit{Reduced-dimension least squares (RD-LS) detector:} 
In the reduced-dimensional model (\ref{restrict}), the matrix $\A_{\hat{\mathcal{I}}}\R_{\hat{\mathcal{I}}}$ introduces interference when we detect the symbols. Borrowing from the idea of conventional MUD decorrelator, we can  alleviate the effect of interference using the method of least-squares, and estimate the symbols by solving $\hat{\bb}_{\hat{\mathcal{I}}} = \arg\min_{\xb}\|\yb - \A_{\hat{\mathcal{I}}}\R_{\hat{\mathcal{I}}}\xb\|^2$. We call this the reduced-dimension least squares (RD-LS) detector. Since $\sign([\hat{\bb}_{\hat{\mathcal{I}}}]_n) = \sign([\R_{\hat{\mathcal{I}}}^2 \hat{\bb}_{\hat{\mathcal{I}}}]_n)$, RD-LS detects symbols by:
\begin{equation}
\hat{b}_n = \left\{
\begin{array}{cc}
\sign\left(r_n \Re \left[(\A_{\hat{\mathcal{I}}}^H \A_{\hat{\mathcal{I}}})^{-1}\A_{\hat{\mathcal{I}}}^H \yb\right]_n \right),& n \in \hat{\mathcal{I}};\\
0, & n\notin \hat{\mathcal{I}}.
\end{array}\right.
\label{LS}
\end{equation}
Similarly, we can modify RDDF by replacing symbol detection by (\ref{LS}) on detected support $\mathcal{I}^{(k)}$ in each iteration.

\subsubsection{\yx{Reduced-Dimension Maximum Likelihood (RD-ML)} detector}

The \yx{RD-ML} detector finds the active users and symbols by solving the following integer optimization problem:
\begin{equation}
\begin{split}
\max_{b_n\in\{-1, 0, 1\}}&
2\yb^H (\A\G^{-1}\A^H)^{-1}\A\R\bb \\
&\qquad - \bb^H \R\A^H(\A\G^{-1}\A^H)^{-1}\A\R\bb,
\end{split}
\label{ML_RDMUD}
\end{equation}
where $b_n = 0$ corresponds to the $n$th user being inactive.
Similar to the conventional maximum likelihood detector, 
\yx{the complexity-per-bit of the RD-ML is exponential in the number of users.  We therefore do not consider this algorithm further.}

\subsection{Choice of $\A$} \label{sec:A}


In Sections \ref{sec:RDD} and \ref{sec:algorithm_II} we have shown that both the RDD and  RDDF detectors are based on the inner products between $\yb$ and the columns of $\A$. \yx{Since $\yb$ consists \ag{of} $\ab_n$ corresponding to the active users \ag{plus noise},} 
intuitively, for RDD and RDDF to work well, the inner products between columns of $\A$, or its coherence as defined in (\ref{def_coherence}), should be small. 
%
Several commonly used random matrices in compressed sensing that have small coherence with high probability are:
\begin{itemize}
    \item[(1)] Gaussian random matrices, where entries $a_{nm}$ are independent and identically distributed (i.i.d.) with a zero mean and unit variance Gaussian distribution, with columns normalized to have unit norm;
    \item[(2)] Randomly sampled rows of a unitary matrix. For instance, the random partial discrete Fourier transform (DFT) matrix, which is formed by randomly selecting rows of a DFT matrix $\vect{F}$: $[\vect{F}]_{nm} = e^{i\frac{2\pi}{N}nm}$ and normalizing the columns of the sub-matrix.
    \item[(3)]  \yaor{Kerdock codes \cite{CalderbankCameronKantor1996}: these codes have dimension restricted to $M\times M^2$, where $M = 2^{m+1}$ with $m$ an odd integer greater than or equal to 3. They have very good coherence properties, with $\mu = 1/\sqrt{M}$ which meets the Welch lower bound on coherence. The Welch bound imposes a general lower bound on the coherence of any $M\times N$ matrix $\A$ \cite{FornasierRauhut2010} leading to $\mu \gtrapprox M^{-1/2}$, when $N$ is large relative to $M$ and $N$ is much larger than 1.}
\end{itemize}

Among these three possible matrix choices,  the random partial DFT matrix has some important properties that simplify closed-form analysis in some cases. In practice, if we choose the number of correlators equal to the number of users, i.e. $M = N$, then there is no dimension reduction, and the performance of RD-MUD should equal that of the MF-bank.
When $M = N$, the random partial DFT matrix becomes the DFT matrix with the property that $\A^H\A = \I$, i.e, $\ab_n^H \ab_m = \delta_{nm}$. \yx{In this case, the set of statistics $\{\ab_n^H\yb\}$ that RDD and RDDF are based on} has the same distribution as the decorrelator output. To see this, write $\ab_n^H \yb = \ab_n^H\left(\sum_{m=1}^N \ab_m r_m b_m\right) + \ab_n^H \wb = r_n b_n + \ab_n^H \wb$, where $\ab_n^H \wb$ is a Gaussian random variable with zero mean and covariance $\sigma^2 \ab_n^H \A\G^{-1}\A^H \ab_m = [\G^{-1}]_{nm}$. 
In contrast, the Gaussian random matrix does not have this property. 
\ye{Therefore, in this setting,} the performance of RD-MUD using a Gaussian random matrix $\A$ is worse than that using the random partial DFT matrix. This is also validated in our numerical results in Section \ref{sec:A_comp}. We will also see in Section \ref{sec:A_comp} that Kerdock codes outperform both random partial DFT and Gaussian random matrices for a large number of users. This is due to their good coherence properties. However, as discussed above, Kerdock codes have restricted dimensions and are thus less flexible for system design.

\section{Performance of RD-MUD}\label{sec:RD-MUD-performance}

We now study the performance of RD-MUD with the RDD and RDDF detectors. \yx{We begin by considering the case of a single active user without noise as a motivating example}. 

\subsection{Single Active User}\label{sec:single_user_det}

When there is only one active user in the absence of noise, the RDD detector can detect the correct active user and symbol by using only \textit{two} correlators, if \yx{every two} columns of $\A$ are linearly independent. Later we will also show this is a corollary (Corollary \ref{corollary_noise_free}) of the more general Theorem \ref{thm_noisy}. 

Assume there is no noise and only one user with index $n_0$ is active. In this case, $y(t) = r_{n_0}b_{n_0}s_{n_0}(t)$ and $K = 1$. 
In RD-MUD, with two correlators, the RDD detector determines the active user by finding
\begin{equation}
\hat{n}_0 = \arg\max_{n = 1,\cdots, N} |a_{1n}\la h_1(t), y(t) \ra + a_{2n}\la h_2(t), y(t) \ra|. \label{eg_RD_MUD}
\end{equation}
From the Cauchy-Schwarz inequality,
\begin{equation}
\begin{split}
 &|a_{1n}\la h_1(t), y(t) \ra + a_{2n}\la h_2(t), y(t) \ra|^2 \\
 &\leq (a_{1n}^2+a_{2n}^2)
 \left[
 \la h_1(t), y(t)\ra^2 +
 \la h_2(t), y(t)\ra^2
 \right],
 \end{split}
\end{equation}
with equality if and only if $a_{mn}=c\la h_m(t), y(t) \ra = c a_{mn_0} r_{n_0} b_{n_0} = c(n_0) a_{mn_0}$ for both $m = 1, 2$ with some constant $c(n_0)$. If \yx{every two} columns of $\A$ are linearly independent, we cannot have two indices $n$ such that $a_{mn} = c(n_0) a_{mn_0}$ for $m = 1, 2$.  Also recall that the columns of $\A$ are normalized, $a_{1n}^2 + a_{2n}^2 = \|\ab_n\|^2 = 1$. Therefore, the maximum is achieved only for $n = n_0$ and $c(n_0) = 1$, which detects the correct active user.
The detected symbol is also correct, since
\begin{equation}
\begin{split}
\hat{b}_{n_0} &= \sign(r_{n_0} [a_{1{n_0}}\la y(t), h_1(t) \ra + a_{2{n_0}}\la y(t), h_2(t)\ra] )\\ &= \sign(r_{n_0}^2 b_{n_0}[a_{1{n_0}}^2 + a_{2{n_0}}^2] ) = b_{n_0}.\end{split}
\end{equation}

\subsection{Noise Amplification of Subspace Projection}\label{noise_amplify}
    
The projection onto the detection subspace amplifies the variance of noise. \yx{When the RDD and RDDF detectors detect the $n$th user, they are affected by a noise component $\ab_n^H\wb$. Consider the special case with orthogonal signature waveforms, i.e. $\G = \I$, and $\A$ chosen as the random partial DFT matrix. In this case, the noise variance is given by $\sigma^2  \ab_n^H\A\A^H\ab_n = \sigma^2(N/M)$, so that it is amplified by a factor $N/M\geq 1$.}
\yx{In general, with subspace projection, the noise variance is amplified by a factor of $\ab_n^H \A\G^{-1}\A^H \ab_n$ \cite{CastroEldar2011}.
Below we  capture this noise amplification more precisely by relating the noise variance to the performance of the RD-MUD detectors.}

\subsection{Coherence Based Performance Guarantee}\label{coh_based}

In this section, we present conditions under which the RDD and RDDF detectors can successfully recover active users and their symbols. The conditions depend on $\A$ through its coherence and are parameterized by the crosscorrelations of the signature waveforms through the properties of the matrix $\G$. Our performance measure is the \ag{probability-of-symbol-error}, which is defined as the \ye{probability} that the set of active users is detected incorrectly, \textit{or} any of their symbols are detected incorrectly:
\begin{equation}
P_{e} = 
\mathbb{P}\{\hat{\mathcal{I}}\neq \mathcal{I}\}+\mathbb{P}\{\{\hat{\mathcal{I}}= \mathcal{I}\}\cap \{\hat{\bb}\neq\bb\}\}.\label{Pe}
\end{equation}
We will show \ye{in the proof of Theorem \ref{thm_noisy}} that the second term of (\ref{Pe}) is dominated by the first term when (\ref{support}) and (\ref{support_OMP}) are used for active user detection. 
%
%
Define the largest and smallest channel gains as
\begin{equation}
|r_{\max}|\triangleq \max_{n}  |r_n|,\quad |r_{\min}|\triangleq \min_{n} |r_n|. \label{gain_def}
\end{equation}
\yao{Also define the $k$th largest channel gain as $|r^{(k)}|$. Hence, $|r_{\max}| = |r^{(1)}|$ and $|r_{\min}| = |r^{(K)}|$.}
Our main result is the following theorem:
\begin{thm}\label{thm_noisy}
Let $\bb \in \mathbb{R}^{N\times 1}$ be an unknown deterministic symbol vector, $b_n \in \{-1, 1\}$, $n\in\mathcal{I}$, and $b_n = 0$, $n\in\mathcal{I}^c$, $n = 1, \cdots, N$. Denote the RD-MUD front-end output by $\y = \A \R \bb + \wb$, where $\A \in \mathbb{C}^{M\times N}$ and $\G \in \mathbb{R}^{N \times N}$ are known, $\wb$ is a Gaussian random vector with zero mean and covariance $\sigma^2 \A\G^{-1}\A^H$ and $\R = \mbox{diag}\{r_1, \cdots, r_N\}$. Let \begin{equation}
\tau \triangleq \sigma \sqrt{2(1+\alpha)\log N} \cdot \sqrt{\lambda_{\max}(\G^{-1})}\cdot \sqrt{\max_n (\ab_n^H \A\A^H \ab_n)}, \label{def_tau}
\end{equation}
for a given constant $\alpha >0$. 
\begin{enumerate}
\item Assume that the number of active users $K$ is known. If the coherence (\ref{def_coherence}) of $\A$ satisfies the following condition: 
\begin{equation}
|r_{\min}| - (2K - 1)\mu |r_{\max}| \geq 2 \tau,\label{cond_thresholding}
\end{equation}
for some constant $\alpha >0$, then the probability-of-symbol-error (\ref{Pe}) for the 
RDD detector is upper bounded by:
\begin{equation}
P_e \leq N^{-\alpha} [\pi(1+\alpha)\log N]^{-1/2}.\label{high_prob_noisy}
\end{equation}
\item \yao{Assume $K_0$ is an upper bound for the number of active users.  If the coherence (\ref{def_coherence}) of $\A$ satisfies (\ref{cond_thresholding}) for $K = K_0$, and we choose a threshold $\xi > 0$ that satisfies
\begin{equation}
K_0 \mu |r_{\max}| + \tau < \xi  < |r_{\min}| - (K_0 - 1)\mu |r_{\max}| - \tau,
\end{equation}
then the probability-of-symbol-error (\ref{Pe}) for the 
RDDt detector is upper bounded by the right-hand-side of (\ref{high_prob_noisy}).}
\item Assume that the number of active users $K$ is known. If the coherence (\ref{def_coherence}) of $\A$  satisfies the following condition:
\begin{equation}
|r_{\min}| - (2K - 1)\mu |r_{\min}| \geq 2 \tau,\label{cond_OMP}
\end{equation}
for some constant $\alpha >0$, then the probability-of-symbol-error (\ref{Pe}) for the RDDF detector is upper bounded by the right-hand-side of (\ref{high_prob_noisy}).
\item \yao{If the coherence (\ref{def_coherence}) of $\A$ satisfies (\ref{cond_OMP}) and we choose a threshold $\epsilon >0$ such that
\begin{equation}
\tau <\epsilon < \min_{k=1}^K \{|r^{(k)}|[1 - (K-k)\mu] - \tau\}, 
\end{equation}
then  the probability-of-symbol-error (\ref{Pe}) for the RDDFt detector is upper bounded by the right-hand-side of (\ref{high_prob_noisy}).}
\end{enumerate}
\end{thm}
\begin{proof}
See Appendix \ref{proof_thm_noisy}.
\end{proof}

The main idea of the proof is the following. We define an event 
\begin{equation}
\mathcal{G} = \left\{\max_{n=1}^N |\ab_n^H \w| < \tau\right\}  \label{event}
\end{equation}
for $\tau$ defined in (\ref{def_tau}), and prove that $\mathcal{G}$ occurs with high probability. This 
bounds the two-sided tail probability of the noise.  
Then we show that under (\ref{cond_thresholding}), whenever $\mathcal{G}$ occurs, the active users can be correctly detected.  On the other hand, we show that under a condition weaker than (\ref{cond_thresholding}), whenever $\mathcal{G}$ occurs, the user data symbols can be correctly detected. A similar but inductive approach is used to prove the performance guarantee for the RDDF detector.

A special case of Theorem \ref{thm_noisy} is when $\A\A^H =(N/M) \I$, $\max_n (\ab_n^H \A\A^H \ab_n) = N/M$ and $\G = \I$. This is true when $\A$ is the random partial DFT matrix and the signature waveforms are orthogonal. If we scale $\sigma^2$ by $M/N$, then the right hand sides of (\ref{cond_thresholding}) and (\ref{cond_OMP}) are identical to the corresponding quantities in Theorem 4 of \cite{Ben-HaimEldarElad2010}. Hence, \yx{in this case} Theorem \ref{thm_noisy} has the same conditions as those of Theorem 4 in \cite{Ben-HaimEldarElad2010}.  However, Theorem 4 in \cite{Ben-HaimEldarElad2010} \yx{only} guarantees detecting the correct sparsity pattern of $\bb$ (equivalently, the correct active users), whereas Theorem \ref{thm_noisy}  guarantees correct detection of not only the active users but their symbols as well. \ye{Theorem \ref{thm_noisy} is also applied to more general colored noise.}

\textit{Remarks:}\label{remarks}
\begin{enumerate}
\item[(1)] The term on the right hand side of  (\ref{cond_thresholding}) and (\ref{cond_OMP})  is bounded by
$1 \leq \max_n (\ab_n^H \A\A^H \ab_n) \leq 1+(N-1)\mu^2.$
\item[(2)] There is a noise phase-transition effect. Define the minimum signal-to-noise ratio (SNR) as  
\begin{equation}
\SNR_{\min} = \frac{|r_{\min}|^2}{\sigma^2 \lambda_{\max}(\G^{-1})},\label{SNR_min}
\end{equation}
where $\lambda_{\max}(\G^{-1})$ captures the noise amplification effect in the subspace projection due to nonorthogonal signature waveforms. Conditions (\ref{cond_thresholding}) and (\ref{cond_OMP}) suggest that for the RDD and RDDF detectors to have $P_e$ as small as (\ref{high_prob_noisy}), we need \yx{at least}
\begin{equation}
\SNR_{\min} > 8\log N.\label{min_SNR_requirement}
\end{equation} 
\yx{This means that} if the minimum SNR is not sufficiently high, then these algorithms cannot attain small probability-of-symbol-error. We illustrate this effect via numerical examples in Section \ref{eg:noise_floor} (a similar effect can be observed in standard MUD detectors). 
\item[(3)] In Theorem \ref{thm_noisy} the condition of having a small probability-of-symbol-error for the RDDF detector is weaker than for the RDD detector. Intuitively, the iterative approach of decision feedback removes the effect of \yx{strong users} iteratively, which helps the detection of weaker users. 
\end{enumerate}

\subsection{Bounding probability-of-symbol-error of RDD and RDDF}

Theorem \ref{thm_noisy} provides a condition on how small $\mu$ has to be to achieve a small probability-of-symbol-error. \yx{We can eliminate the constant $\alpha$ and rewrite} Theorem \ref{thm_noisy} in an equivalent form that gives explicit error bounds for the RDD and RDDF detectors. 
Define
\begin{equation}
\begin{split}
&\beta_1 \triangleq \frac{[1-(2K-1)\mu |r_{\max}|/|r_{\min}|]^2}{\max_n (\ab_n^H \A\A^H \ab_n)},\\
&\beta_2 \triangleq \frac{[1-(2K-1)\mu]^2}{\max_n (\ab_n^H \A\A^H \ab_n)}. 
\end{split}\label{new_beta}
\end{equation}
For the RDD detector, we have already implicitly assumed that $1-(2K-1)\mu |r_{\max}|/|r_{\min}| \geq 0$, since the right hand side of (\ref{cond_thresholding}) in Theorem \ref{thm_noisy} is non-negative. For the same reason, for the RDDF detector, we have assumed that $1-(2K-1)\mu > 0$. By \yx{Remark (1)} and (\ref{new_beta}), $\beta_1 \leq 1$ and $\beta_2 \leq 1$. We can prove the following corollary from Theorem \ref{thm_noisy} (see \cite{Xie2011PhD} for details):
\begin{corollary}
Under the setting of Theorem \ref{thm_noisy}, with the definitions (\ref{SNR_min}) and (\ref{new_beta}), the probability-of-symbol-error for the RDD detector is upper-bounded by
\begin{equation}
P_{e, {\rm RDD}} \leq  \frac{2 N}{\sqrt{\pi}} 
\left[\frac{\SNR_{\min}}{2} \cdot \beta_1 \right]^{-1/2}
e^{-\frac{1}{4}\frac{\SNR_{\min}}{2} \cdot \beta_1}, \label{form_thresholding}
\end{equation}
with $1-(2K-1)\mu |r_{\max}|/|r_{\min}| \geq 0$,
and the probability-of-symbol-error for the RDDF detector is upper bounded by
\begin{equation}
P_{e, {\rm RDDF}} \leq  \frac{2 N}{\sqrt{\pi}} 
\left[\frac{\SNR_{\min}}{2} \cdot \beta_2 \right]^{-1/2}
\cdot
e^{-\frac{1}{4}\frac{\SNR_{\min}}{2} \beta_2 }, \label{form_OMP}
\end{equation}
with  $1-(2K-1)\mu > 0$. 
\end{corollary}

For the decorrelating detector in conventional MUD, a commonly used performance measure is the probability of error of each user \cite{LupasVerdu1989}, which is given by:
\begin{equation}
\mathbb{P}\{\hat{b}_n \neq b_n\} = Q\left(|r_n|/(\sigma\sqrt{[\G^{-1}]_{nn}})\right), \label{known}
\end{equation}
where $Q(x) = \int_x^\infty (1/\sqrt{2\pi}) e^{-z^2/2}dz$ is the Gaussian tail probability.  Using (\ref{known})  and the union bound, we obtain
\begin{align}
P_e  &= \mathbb{P}\{\hat{\bb}\neq \bb\}  \leq \sum_{n=1}^N \mathbb{P}\{\hat{b}_n \neq b_n\} 
 \leq  N Q\left(\sqrt{{\rm SNR}_{\min}}\right) \nonumber\\
& \leq
 \frac{N}{2\sqrt{\pi} }
\left[\frac{{\rm SNR}_{\min}}{2}\right]^{-1/2} e^{-\frac{{\rm SNR}_{\min}}{2}},
 \label{MF_decorrelator_exist}
\end{align} 
where we also used the fact that $|r_n|/[\sigma\sqrt{[\G^{-1}]_{nn}}] \geq   \sqrt{{\rm SNR}_{\min}}$ and $Q(x)$ is decreasing in $x$, as well as the bound \cite{verduMUD1998} 
$
Q(x) \leq 1/(x\sqrt{2\pi})e^{-x^2/2}. \label{Q_bound}
$ 
Since conventional MUD is not concerned with active user detection and the errors are due to symbol detection, it only makes sense to compare  (\ref{MF_decorrelator_exist})  to  (\ref{form_thresholding}) and (\ref{form_OMP}) when $K = N$ and $M = N$. Under this setting, $\beta_1 = 1$ and $\beta_2 = 1$, and the bounds on $P_e$ (\ref{form_thresholding}) of RDD and (\ref{form_OMP})  of RDDF are larger than the bound (\ref{MF_decorrelator_exist}) of the decorrelating detector. 
This is because when deriving bounds for symbol detection error in the proof of Theorem \ref{thm_noisy}, we consider the probability of (\ref{event}), which requires two-side tail-probability of a Gaussian random variable. In contrast, in conventional MUD, without active-user detection, only the one-sided tail probability of the Gaussian random variable $\mathbb{P}\{\wb > \tau\}$ is required \yaor{because we use binary modulation}. Nevertheless, obtaining a tighter bound for symbol detection error is not necessary in RD-MUD because when $K<N$, active user detection error dominates symbol detection error.

By letting the noise variance $\sigma^2$ go to zero in (\ref{form_thresholding}) and (\ref{form_OMP}) for the RDD and RDDF detectors, we can derive the following corollary from Theorem \ref{thm_noisy} (\ag{a proof of this corollary for the RDD detector has been given in Section \ref{sec:single_user_det}). }
%
\begin{corollary}\label{corollary_noise_free}
Under the setting of Theorem \ref{thm_noisy}, in the absence of noise, the RDD detector can correctly detect the active users and their symbols if $\mu < |r_{\min}|/[|r_{\max}|(2K-1)]$, and the RDDF detector can correctly detect the active users and their symbols if $\mu < 1/(2K-1)$.
In particular, if $K = 1$, with $M = 2$ correlators, $P_e = 0$ for the RDDF detector, and if furthermore $|r_{\max}| = |r_{\min}|$, $P_e = 0$ for the RDD detector (which has also been shown in Section \ref{sec:single_user_det}).
\end{corollary}

\subsection{Lower Bound on the Number of Correlators}

Theorem \ref{thm_noisy} is stated for any matrix $\A$. By substitution of the expression for coherence of a given $\A$ in terms of its dimensions $M$ and $N$ into Theorem \ref{thm_noisy}, we can obtain a lower bound on the smallest number of correlators $M$ needed to achieve a certain probability-of-symbol-error. 
%
For \yx{example, the coherence of the random partial DFT matrix can be bounded in probability (easily provable} by the complex Hoeffding's inequality \cite{Hoeffding1963}):
\begin{lemma}\label{lemma_partial_DFT}
Let $\A \in \mathbb{C}^{M \times N}$ be a random partial DFT matrix. Then the coherence of $\A$ is bounded by
\begin{equation}
\mu < \left[4(2\log N + c)/M \right]^{1/2},
\end{equation}
with probability exceeding $1- 2e^{-c}$, for some constant $c>0$.
\end{lemma}
Lemma \ref{lemma_partial_DFT} \yx{together with Theorem \ref{thm_noisy} imply that for the partial DFT matrix to attain a small probability-of-symbol-error, the number of correlators needed by the RDD and RDDF detectors is on the order of $\log N$. }This is much smaller than that required by the conventional MUD using a MF-bank, which is on the order of $N$.


Corollary \ref{corollary_noise_free} together with the Welch bound imply that, for the RDD and RDDF detectors to have perfect detection, the number of correlators $M$ should be on the order of $(2K-1)^2$. In the compressed sensing literature, it is known that the bounds obtained using the coherence property of a matrix have a ``quadratic bottleneck'' \cite{FornasierRauhut2010}: the number of measurements is on the order of $K^2$. Nevertheless, the coherence property is easy to check for a given matrix, and it is a  convenient measure of the user interference level in the detection subspace  as we demonstrated in the proof of Theorem \ref{thm_noisy}.

\section{Numerical Examples}\label{sec:numerical_eg}

As an illustration of the performance of RD-MUD, we present some numerical examples.  \ye{We first generate $10^5$ partial random DFT matrices and choose the matrix that has the smallest coherence as $\A$. Then using the fixed $\A$, we obtain results from $5\times 10^5$ Monte Carlo trials.} For each trial, we generate a Gaussian random noise vector and random bits:  $b_n \in  \{-1, 1\}$, $n\in \mathcal{I}$ with probability 1/2. In this setting, the conventional decorrelating detector has $P_e$ equal to that of the RDD when $M = N$. 

\subsubsection{$P_e$ vs. $M$, as $N$ increases}

Fig. \ref{Fig:comp} shows the $P_e$ of the RDD, RDDF, RDDt and RDDFt detectors as a function of $M$, for fixed $K = 2$, and different values of $N$.
The amplitudes $r_n = 1$ for all $n$, the noise variance is $\sigma^2 = 0.005$, and $\G = \I$, which corresponds to $\SNR _{\min}= 23$dB.  \yaor{For each combination of $N$ and $M$, we numerically search to find the best values for parameters $\xi$ and $\epsilon$. The values of $\xi$ range from 0.78 to 0.92, increase for larger $N$ and decrease for larger $M$ for the RDDt detector.  The values of $\epsilon$ range from 0.50 to 0.80, increase for larger $N$ and decrease for larger $M$ for the RDDFt detector.} 
The RDD and RDDF detectors can achieve small $P_e$ for $M$ much smaller than $N$; the RDDt and RDDFt have some sacrifice in performance due to their lack of knowledge of $K$. This degradation becomes more pronounced for larger values of $N$. 

\begin{figure}[h!]
\begin{center}
\subfloat[RDD and RDDt]{\label{Fig:grouse}
\includegraphics[width = .8\linewidth]{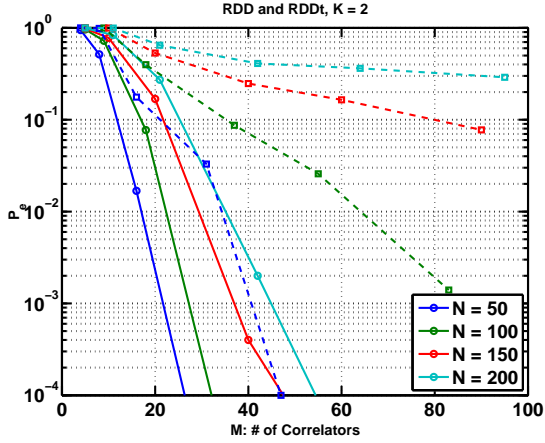}}
\qquad
\subfloat[RDDF and RDDFt]
{\label{Fig:petrels}
\includegraphics[width = .8\linewidth]{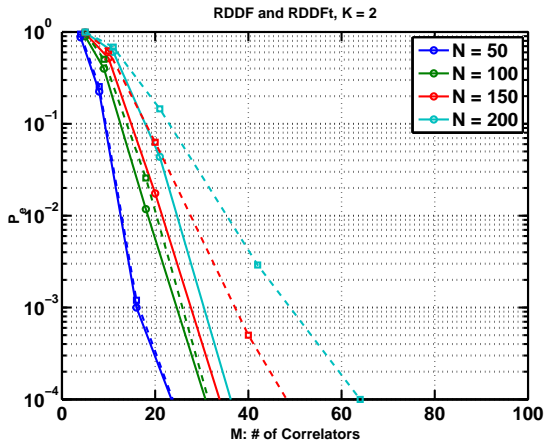}}
\end{center}
\caption{$P_e$ versus $M$  for  $K = 2$ and different $N$. The amplitudes $r_n = 1$ for all $n$, the noise variance is $\sigma^2 = 0.005$, and $\G = \I$, which corresponds to $\SNR _{\min}= 23$dB. (a) RDD and RDDt, where the solid lines correspond to RDD and the dashed lines correspond to RDDt; and (b) RDDF and RDDFt, where the solid lines correspond to RDDF and the dashed lines correspond to RDDFt. }
\label{Fig:comp}
\end{figure}

\subsubsection{$P_e$ vs $M$, as $K$ increases}

Fig. \ref{Fig:MvsK} demonstrates the $P_e$ of the RDD, RDDF, RDDt, RDDFt detectors as a function of $M$, for a fixed $N = 100$, and different values of $K$.  \yaor{For each combination of $M$ and $K$, we numerically search to find the best values for the parameters $\xi$ and $\epsilon$. Here $\xi$ ranges from 0.68 to 0.80 and $\epsilon$ ranges from 0.32 to 0.70.}  The amplitudes $r_n = 1$ for all $n$, the noise variance is $\sigma^2 = 0.005$, and $\G = \I$, which corresponds to $\SNR _{\min}= 23$dB. 
Clearly, the number of correlators needed to obtain small $P_e$ increases as $K$ increases.

\begin{figure}[h]
\begin{center}
\subfloat[RDD and RDDt]{\label{Fig:MvsK}
\includegraphics[width = .8\linewidth]{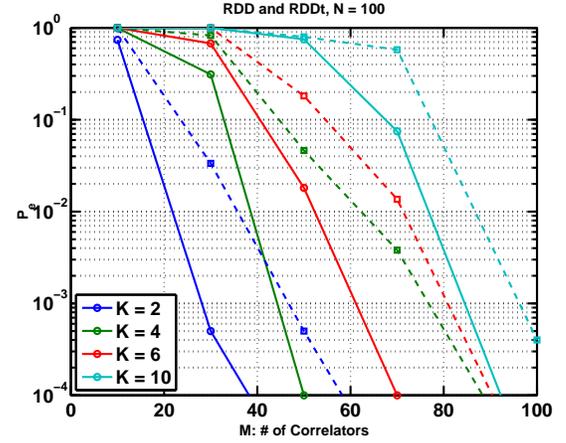}}
\qquad
\subfloat[RDDF and RDDFt]
{\label{Fig:Pe_OMP}
\includegraphics[width = .8\linewidth]{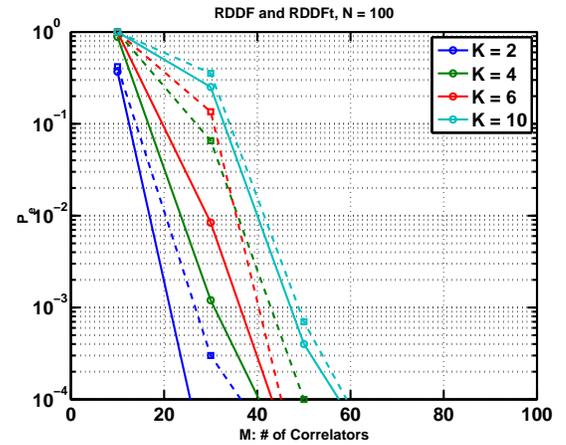}}
\end{center}
\caption{$P_e$ versus $M$ for $N = 100$ and different $K$. The amplitudes $r_n = 1$ for all $n$, the noise variance is $\sigma^2 = 0.005$, and $\G = \I$. (a) RDD and RDDt, where the solid lines correspond to RDD and the dashed lines correspond to RDDt; and (b) RDDF and RDDFt, where the solid lines correspond to RDDF and the dashed lines correspond to RDDFt.}
\end{figure}

\subsubsection{Comparison of random matrices $\A$}\label{sec:A_comp}

\yaor{We compare the $P_e$ of the RDD and RDDF detectors when the Gaussian random matrices, the random partial DFT matrices or the Kerdock codes are use for $\A$.} 
In Fig. \ref{Fig:DFT_vs_Gaussian}, the $P_e$ of the Gaussian random matrix converges to a value much higher than that of the partial DFT matrix, when $M$ increases to $N$. In this example, $N = 100$, $K = 6$, the amplitudes $r_n = 1$ for all $n$, the noise variance is $\sigma^2 = 0.005$, and $\G = \I$. \yaor{In Fig. \ref{Fig:kerdock_comp}, Kerdock codes outperform both the partial DFT and the Gaussian random matrices because of their good coherence properties. This behavior can be explained as follows.  For larger $N$ and relatively small $M$, it becomes harder  to select a partial DFT matrix with small coherence by random search, whereas the Kerdock codes can be efficiently constructed and they obtain the Welch lower bound on coherence by design. For fixed $N$ and $M$, Kerdock codes can support a large number of active users, as demonstrated in Fig. \ref{Fig:Kerdock}. In this example, the coherence of the Kerdock code is $\mu = 0.0312$, which is much smaller than the coherence $\mu = 0.0480$ obtained by choosing from $10^5$ random partial DFT matrices.} Kerdock codes are tight frames \cite{Christensen2003,CalderbankCameronKantor1996} meaning that $\G = N/M \I$  so that no pre-whitening is needed. 

\begin{figure}[h!]
\begin{center}
\subfloat[DFT vs. Gaussian]{\label{Fig:DFT_vs_Gaussian}
\includegraphics[width = .8\linewidth]{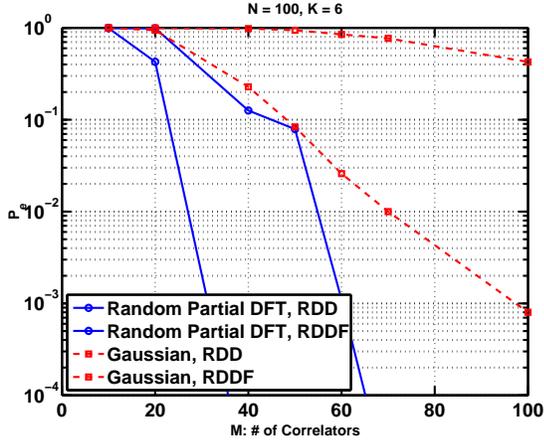}}
\qquad
\subfloat[DFT vs. Gaussian vs. Kerdock]
{\label{Fig:kerdock_comp}
\includegraphics[width = .8\linewidth]{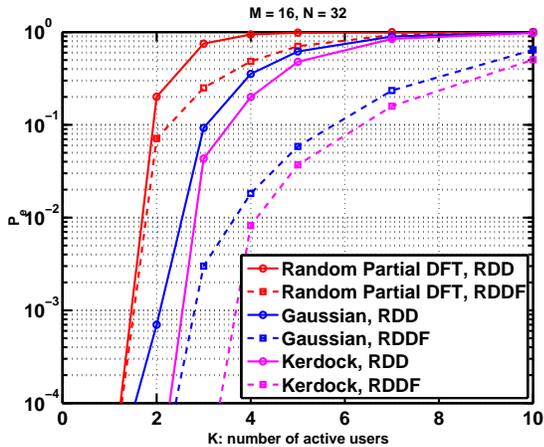}}
\end{center}
\qquad
\subfloat{}
\caption{(a): $P_e$ versus $M$ of the RDD and RDDF detectors using random partial DFT versus Gaussian random matrices for $N = 100$ and $K = 6$ (Kerdock codes require dimensions of $M = 2^m$ for $m = 4, 6, \ldots$ and hence are not presented here). (b) $P_e$ versus $K$ of the RDD and RDDF detectors using Gaussian random matrices, random partial DFT matrices, and Kerdock codes of size 16 by 256 (arbitrarily select 32 columns  for 32 users), for $N = 32$ and $M = 16$. In both examples, the amplitudes $r_n = 1$ for all $n$, the noise variance is $\sigma^2 = 0.005$, and $\G = \I$.}
\label{Fig:comp}
\end{figure}

\begin{figure}[h]
\centering{
\includegraphics[width = .8\linewidth]{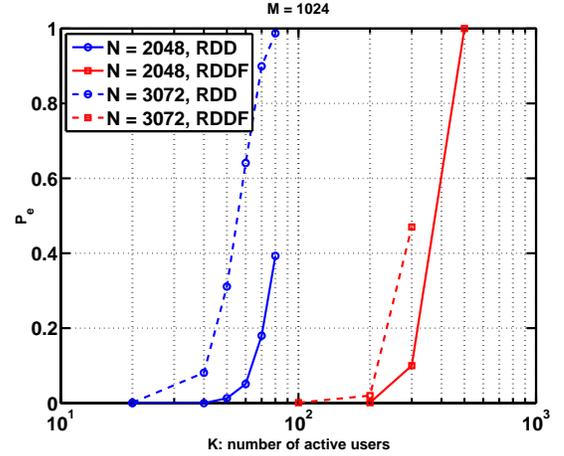}
}
\vspace{0.1in}
\caption{Performance of RDD and RDDF detectors, when using Kerdock codes for $\A$ with $M = 1024$, $P_e$ versus $K$ for various $N$ when  amplitudes $r_n$ uniformly random in [1, 1.5], $\G = \I$ and $\sigma^2 =0.005$. }
\label{Fig:Kerdock}
\end{figure}

\subsubsection{$P_e$ vs. $M$, as $\SNR$ changes}\label{eg:noise_floor}

Consider a case where $\SNR_{\min}$ changes by fixing $\G = \I$, $r_n = 1$ for all $n$, and varying $\sigma^2$. For comparison, we also consider the conventional decorrelating detector, which corresponds to the RDD detector with $M = N$. Assume $N = 100$ and $K = 2$. Note that there is a noise phase-transition effect in Fig. \ref{Fig:MvsK}, which is discussed in the Remarks of Section \ref{coh_based}. 

\begin{figure}[h]
\begin{center}
\subfloat[RDD]{
\includegraphics[width = .8\linewidth]{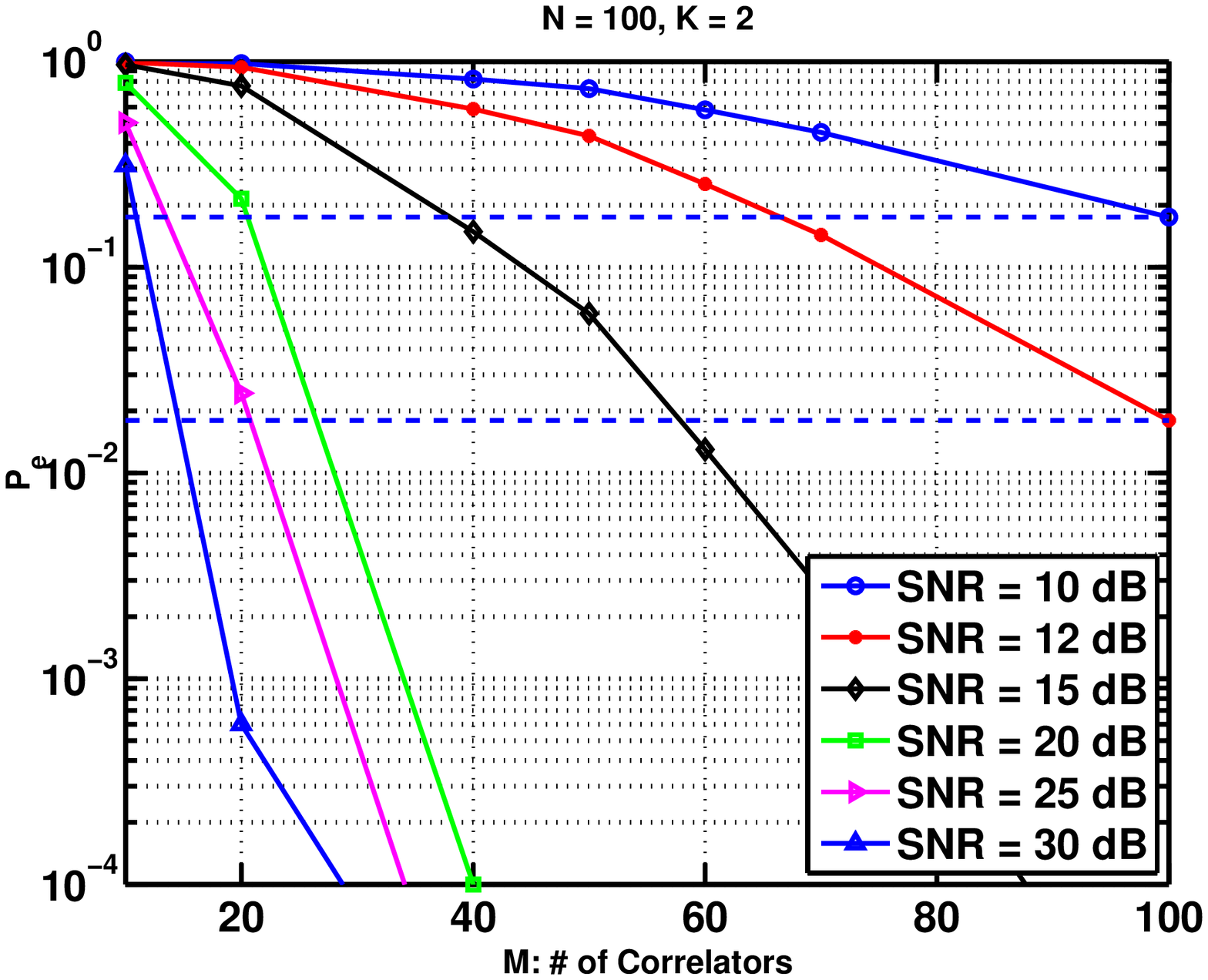}}
\qquad
\subfloat[RDDF]
{
\includegraphics[width = .8\linewidth]{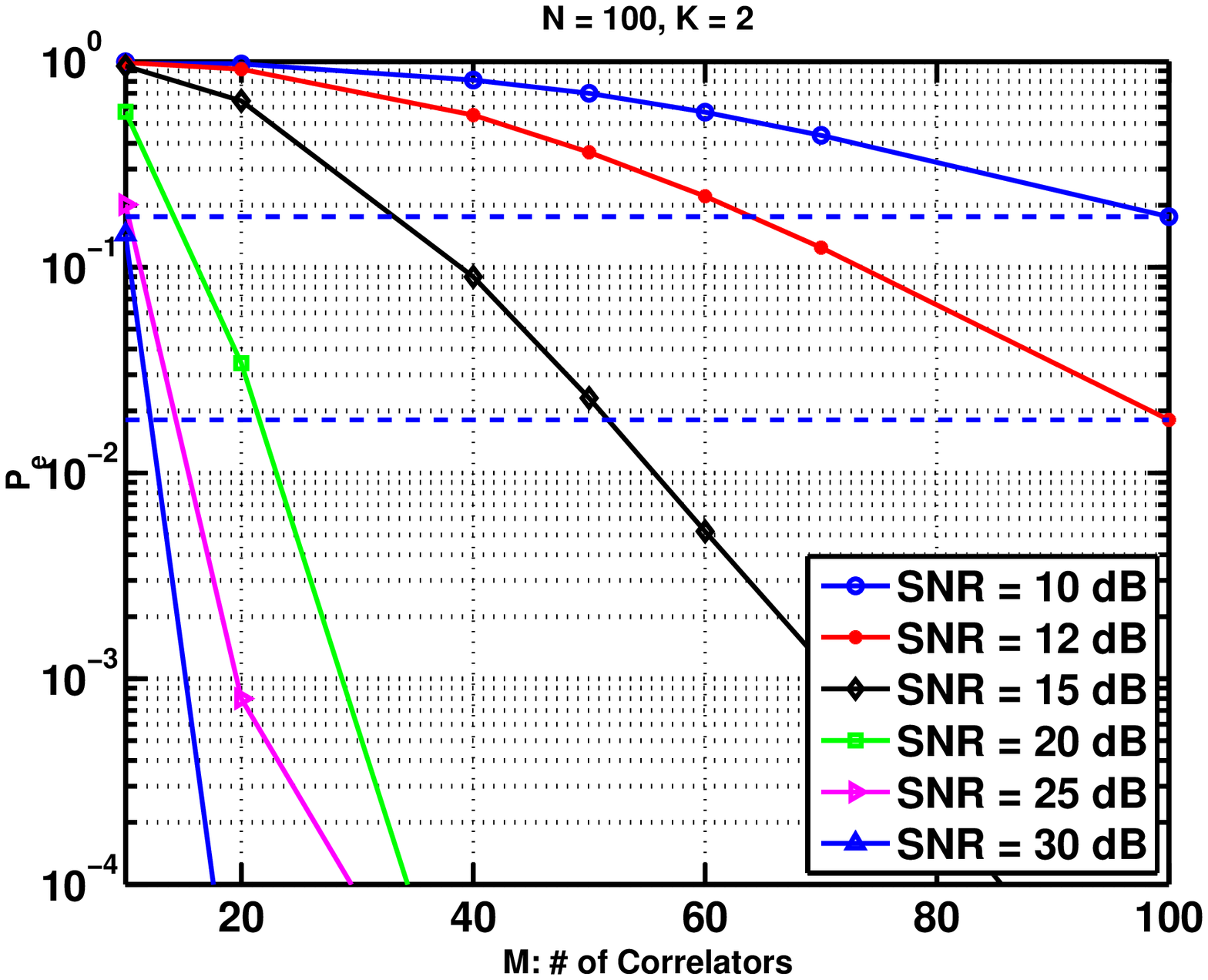}}
\end{center}
\caption{Performance of  RDD and RDDF detectors when $r_n = 1$ for all $n$, $\G=\I$ and various $\sigma^2$, where we denote $\SNR =10\log_{10} (r_n^2/\sigma^2)$dB. The dashed lines show $P_e$ for the conventional decorrelating detectors at the corresponding $\SNR$. 
}
\label{Fig:Pe_SNR}
\end{figure}

\subsubsection{Near-far problem, $\G = \I$}\label{eg:noise_floor}

To illustrate the performance of the RDD and RDDF detectors in the presence of the near-far problem, we choose $r_n$ uniformly random from $[1, 1.5]$ for active users. Assume $N = 100$, $K = 2$, $\sigma^2 = 0.005$. In Fig. \ref{Fig:near_far}, RDDF significantly outperforms RDD.

    \begin{figure}[h]
    \centering
        \includegraphics[width=.8\linewidth]{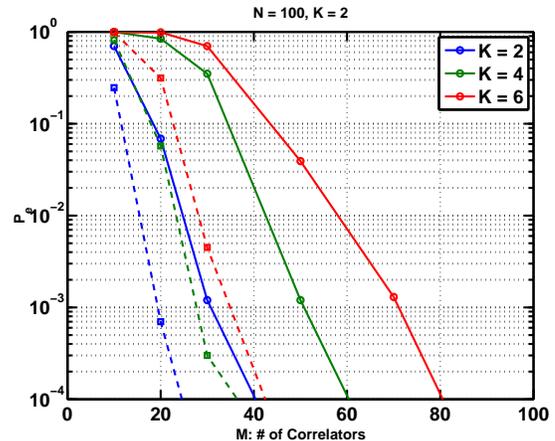}
        \caption{Comparison of RDD and RDDF in the presence of near-far problem, with amplitudes $r_n$ uniformly random in [1, 1.5], $N = 100$, $K = 2$, $\sigma^2 = 0.005$, and $\G = \I$. The solid lines correspond to RDD and the dashed lines correspond to RDDF.}
        \label{Fig:near_far}
    \end{figure}


\subsubsection{$P_e$ vs. $M$, performance of the noise whitening transform}\label{eg:prew}

Next we consider practical signature waveforms in CDMA systems. \yaor{There are many choices for signature sequences and the Gold code is one that is commonly used \cite{Goldsmith2005}. For signature sequences $\{s_{n\ell}\}$,} the signature waveforms are generated by $s_n(t) = \sum_{\ell=0}^{L-1} s_{n\ell} p(t-\ell T_c)$, where $L$ is the sequence length, $T_c\ll T$ is the chip duration, and the sequences $\{s_{n\ell}\}$ are modulated by unit-energy chip waveform $p(t)$ with $ \int |p(t)|^2 dt = 1$ and $\int p(t-\ell T_c) p(t- kT_c) dt = 0$, $\ell \neq k$.  
For Gold codes, we choose $m = 10$ (with length $L = 2^{10} - 1 = 1023$ and 1025 possible codewords) \cite{WangGeICC2008}. We use $100$ Gold codes to support $N = 100$ users.  The Gram matrix of the Gold code is given by 
\begin{equation}
\G = \frac{L+1}{L} \I_{N\times N} - \frac{1}{L} \vect{1}\vect{1}\transpose, \label{Gold_G}
\end{equation}  
which has two distinct eigenvalues. In this example, $\lambda_1 = (N+1)/N = 1.0010$, $\lambda_2 = (L-N+1)/L = 0.8768$, $\lambda_{\max}(\G^{-1}) = 1.1405$ and hence the signature waveforms are nearly orthogonal. \yaor{We also consider  a simulated $\G = \vect{U} \mbox{diag}\{1/400, 2/400, \cdots, 100/400\}\vect{U} \transpose$ for a randomly generated unitary matrix $\vect{U}\in \mathbb{R}^{100\times 100}$, and hence $\lambda_{\max}(\G^{-1}) = 400$ which is much larger than that of the Gold codes.}
%
In Fig. \ref{Fig:Gold_largeSNR} and Fig. \ref{Fig:Gold_smallSNR}, when the signature waveforms are nearly orthogonal, the noise whitening transform does not reduce $P_e$ much. 
%
Fig. \ref{Fig:G_largeSNR} and Fig. \ref{Fig:Gold_smallSNR} show that the performance of the RDD and RDDF detectors can be significantly improved by the noise whitening transform for large $M$. 
We also verified that using the noise whitening transform cannot achieve the probability-of-error that is obtained with orthogonal signature waveforms $\G = \I$. This is because the noise whitening transform distorts the signal component.


\begin{figure}[h]
\begin{center}
\subfloat[Gold code, $\sigma^2 = 0.005$]{
\includegraphics[width = .4\linewidth]{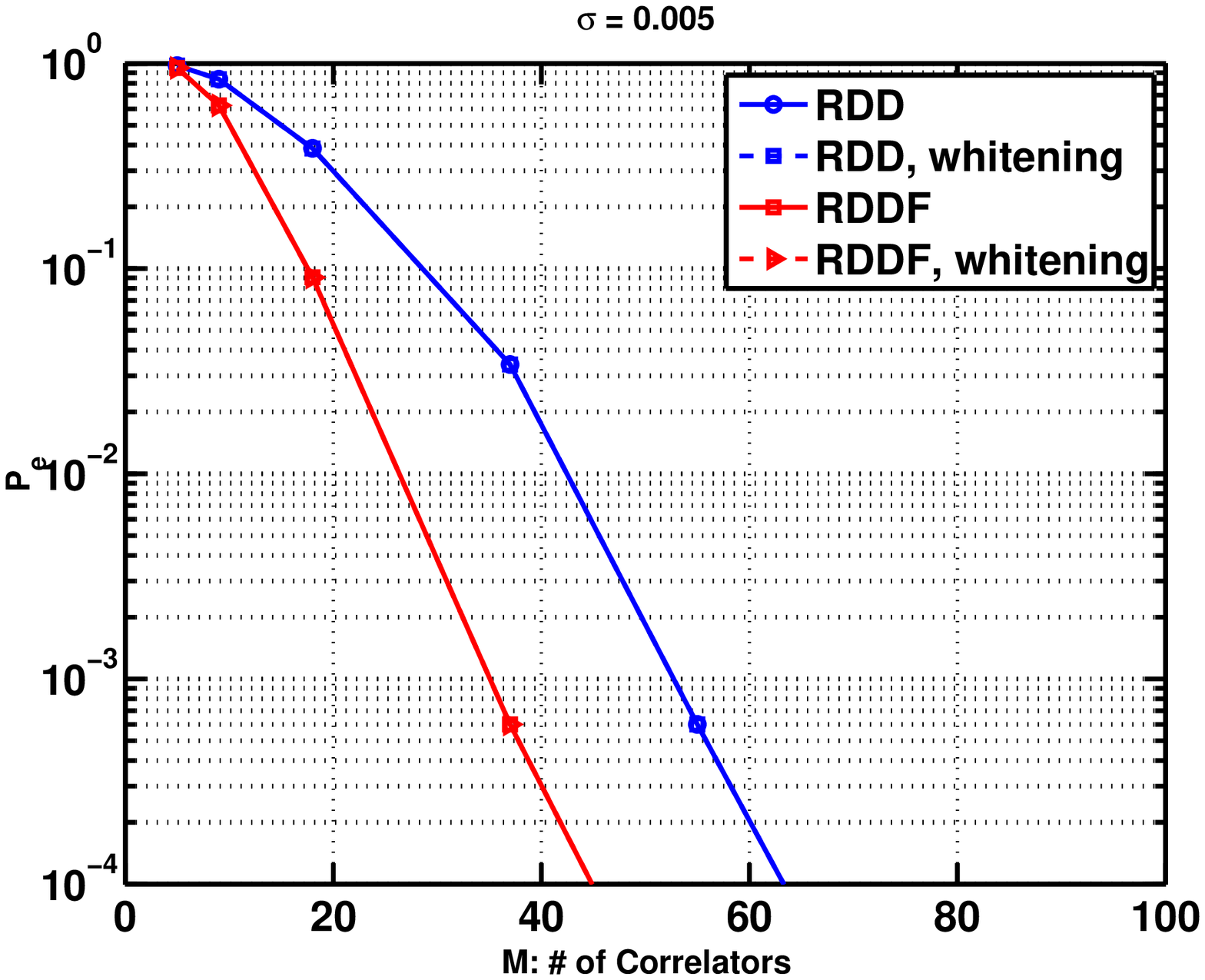}
\label{Fig:Gold_largeSNR}}
\quad
\subfloat[Gold code, $\sigma^2 = 0.01$]
{
\includegraphics[width = .4\linewidth]{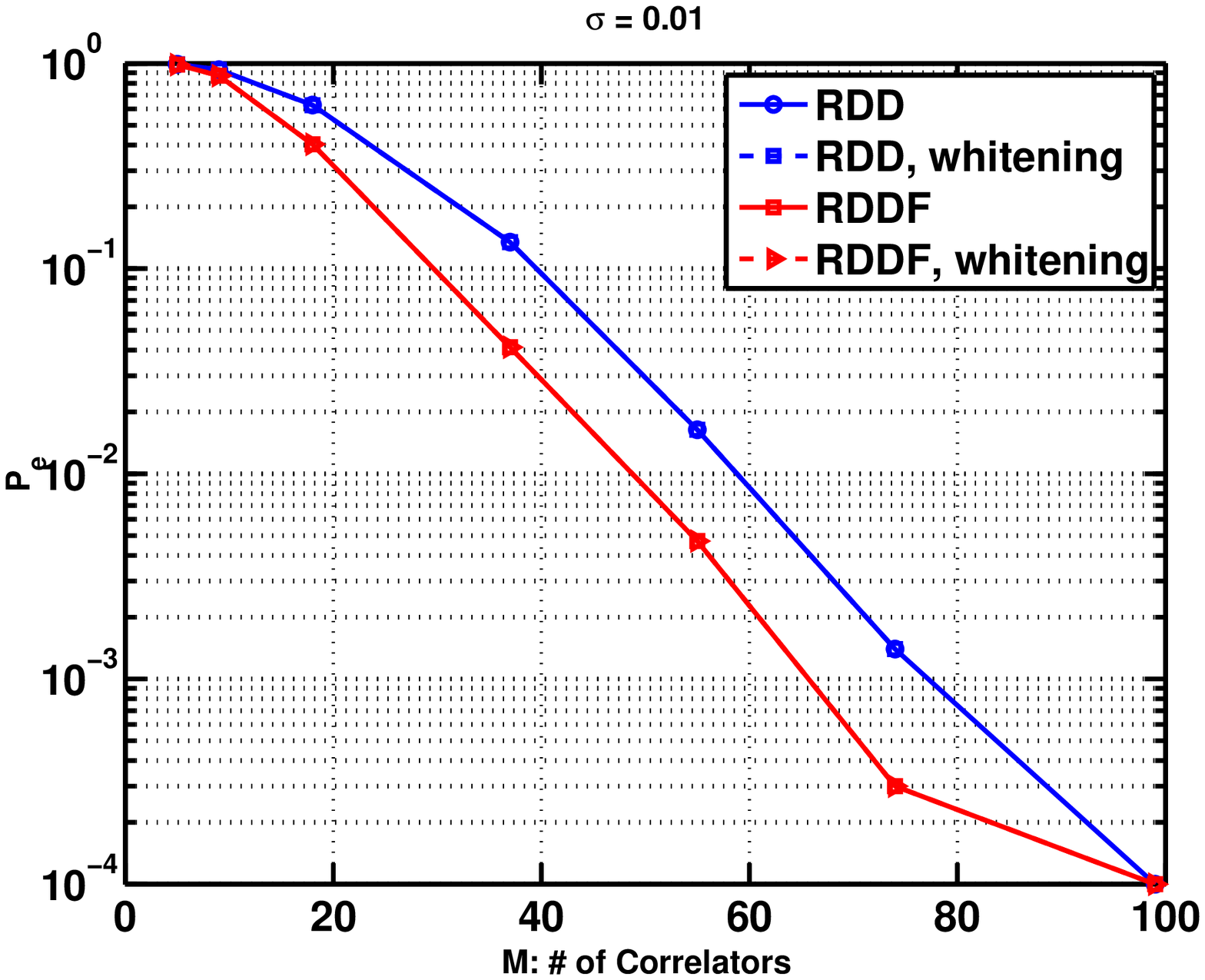} \label{Fig:Gold_smallSNR}}\\
\subfloat[Simulated $\G$, $\sigma^2 = 0.005$]{
\includegraphics[width = .4\linewidth]{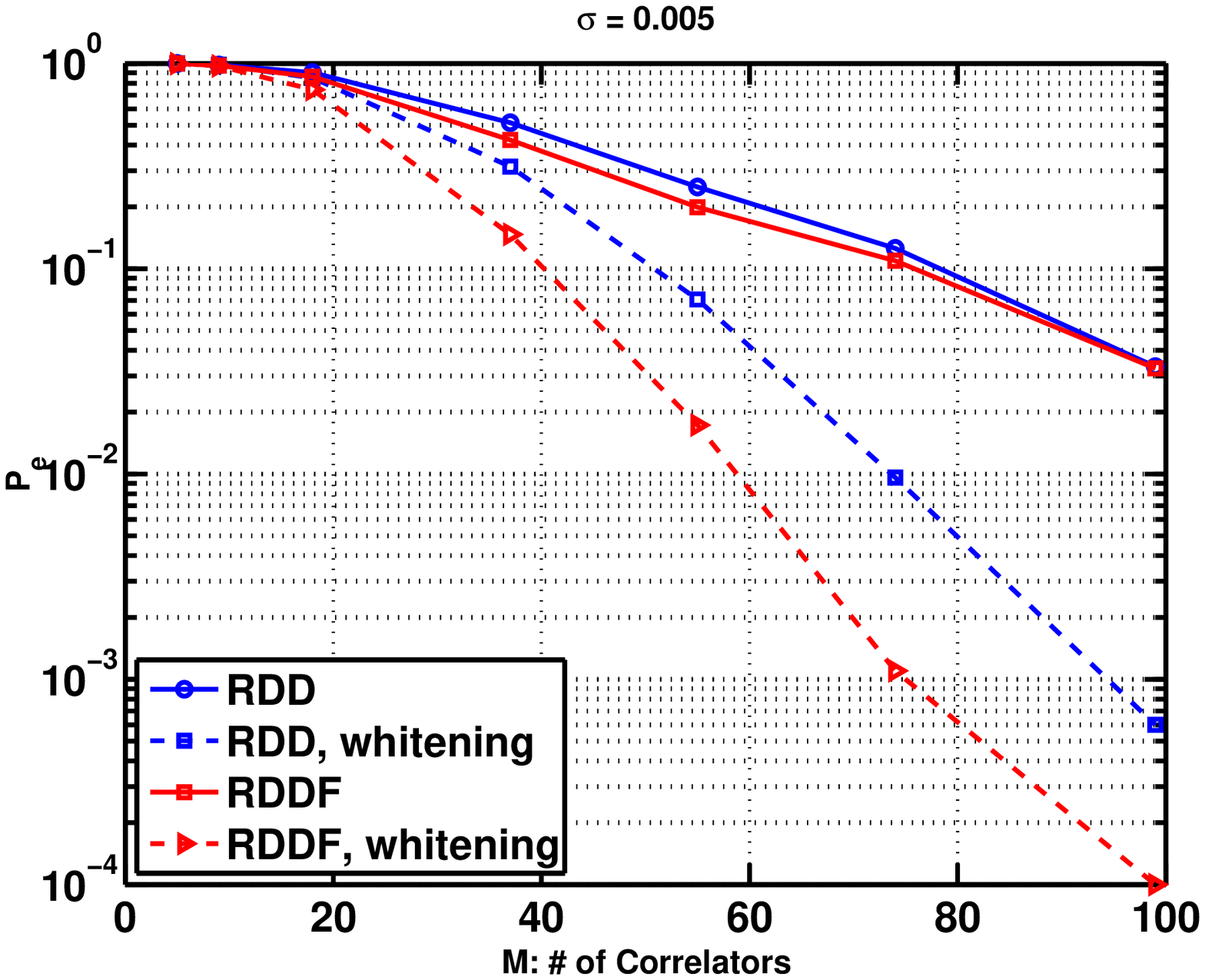}
\label{Fig:G_largeSNR}
}
\quad
\subfloat[Simulated $\G$, $\sigma^2 = 0.01$]
{
\includegraphics[width = .4\linewidth]{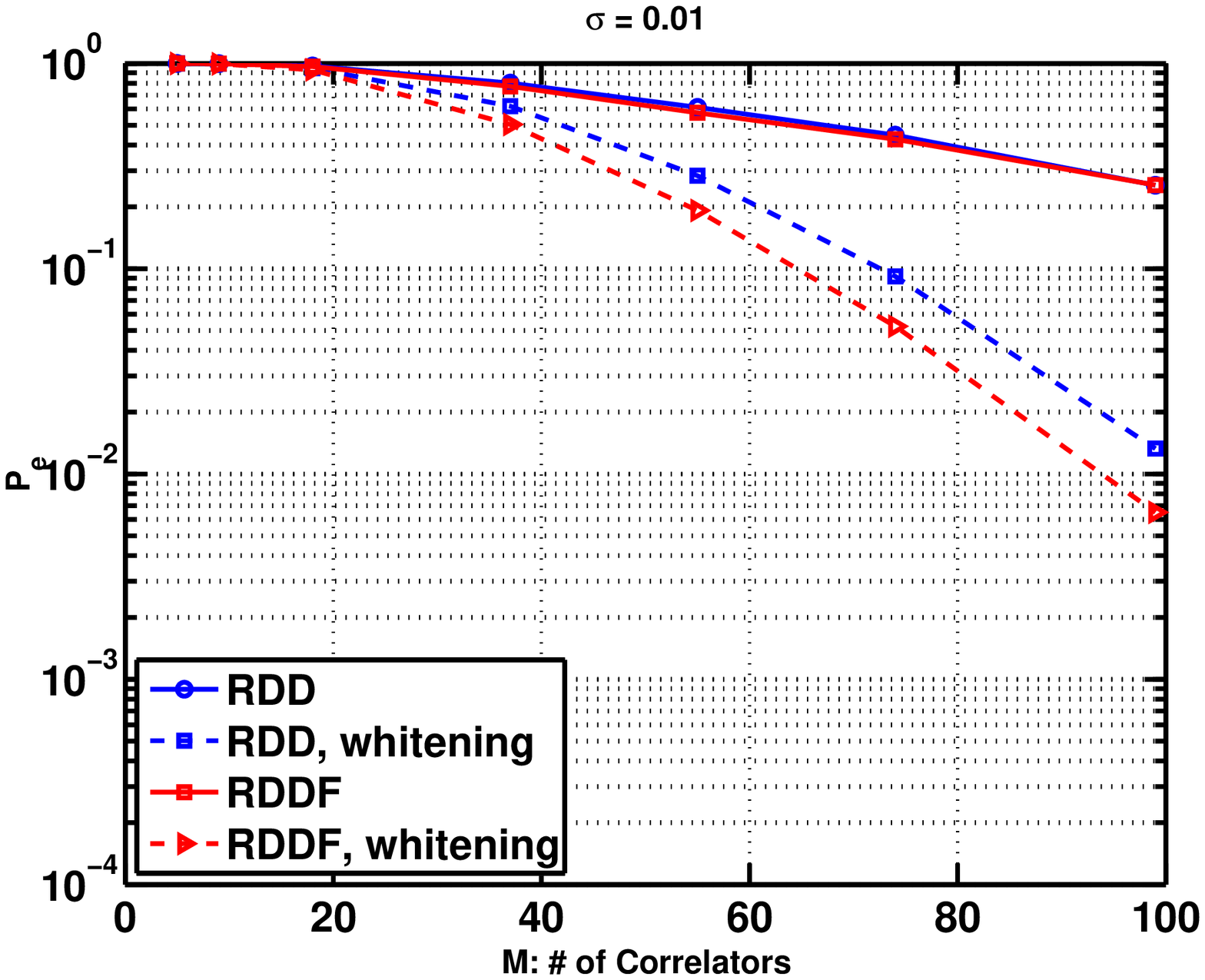}
\label{Fig:G_smallSNR}}
\end{center}
\caption{Comparison of RDD and RDDF detectors with and without noise whitening when $N = 100$, $K = 2$,  amplitudes $r_n$ uniformly random in [1, 1.5], and the following settings for $\G$ and $\sigma^2$: (a) Gold codes with $\lambda_{\max}(\G^{-1}) = 1.1405$, $\sigma^2 =0.005$, (b) same Gold codes as in (a) but $\sigma^2 = 0.01$, (c) simulated $\G$ with $\lambda_{\max}(\G^{-1}) = 400$, $\sigma^2 = 0.005$, (d) same simulated $\G$ as in (c) but $\sigma^2 = 0.01$.
}
\label{Fig:prew}
\end{figure}

\subsubsection{$P_e$ vs. $M$, RD-MUD linear detectors}\label{sec:eg_linear_det}
    
To compare performance of the RD-MUD linear detectors, we consider two sets of schemes. The first are one-step methods: using (\ref{support}) for active user detection followed by symbol detection using (\ref{RD-MUD-sign}) (corresponds to RDD), (\ref{MMSE_RDMUD}) (corresponds to RD-MMSE), or (\ref{LS}) (corresponds to RD-LS), respectively. The second set of schemes detects active users and symbols iteratively: \yaor{the RDDF detector, the modified RDDF detector, modified by replacing the symbol detection by the RD-LS detector (\ref{LS}) on the detected support in each iteration $\mathcal{I}^{(k)}$, and the modified RDDF detector, modified by replacing the symbol detection by the MMSE detector (\ref{MMSE_RDMUD}) on the detected support in each iteration $\mathcal{I}^{(k)}$.} Assume $N = 100$, $K = 2$, $r_n = 1$ for all $n$, and $\sigma^2 = 0.005$. Again we consider Gold codes as defined in Section \ref{eg:prew}. As showed by Table \ref{Table1}, iterative methods including RDDF outperform the one-step methods including RDD. However, the difference between various symbol detection methods is very small, since  active user detection error dominates the symbol detection error. By examining the conditional probability-of-symbol-error $\mathbb{P}\{\hat{\bb}\neq \bb|\hat{\mathcal{I}}=\mathcal{I}\}$, in Fig. \ref{Fig:linear_det} we see that both RD-LS and RD-MMSE detectors have an advantage over sign detection. 

\begin{figure}[htb]
    \centering
        \includegraphics[width=.8\linewidth]{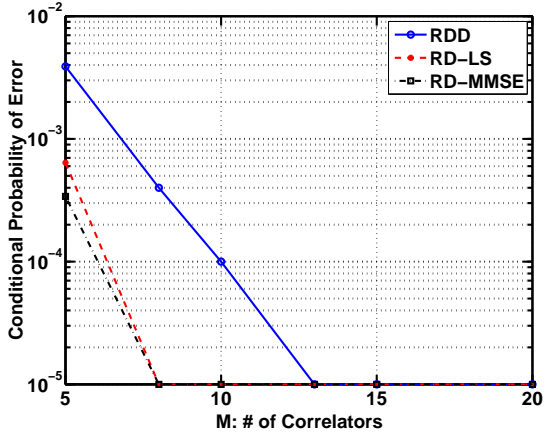}
            \caption{Comparison of $\mathbb{P}\{\hat{\bb}\neq \bb|\hat{\mathcal{I}}=\mathcal{I}\}$ for RDD, RD-LS, and RD-MMSE in the same setting as that in Fig. \ref{Fig:Gold_largeSNR}. }
        \label{Fig:linear_det}
            \end{figure}

\begin{table}[h]
\caption{$\mathbb{P}\{\hat{\bb}\neq \bb|\hat{\mathcal{I}}=\mathcal{I}\}$ vs. $M$, $N = 100$, $K = 2$}
\begin{center}
\begin{tabular}{p{3cm}|lllll}
\hline
& \multicolumn{4}{c}{$M$}  \\\hline
Methods & 5 & 9 & 18 & 37  \\\hline
RDD & 0.9780  &  0.8400 &   0.3857  &  0.0342  \\\hline
RD-LS &  0.9780   & 0.8400  &  0.3857 &   0.0342   \\\hline
RD-MMSE & 0.9779  &  0.8400 &  0.3857 &   0.0342   \\\hline
RDDF &0.9527 &   0.6248   & 0.0905&    0.0006  \\\hline
modified RDDF with LS &  0.9526  &  0.6247   & 0.0905   &  0.0006 \\\hline
modified RDDF with MMSE & 0.9526  & 0.6247  &  0.0905   & 0.0006  \\\hline
\end{tabular}
\end{center}
\label{Table1}
\end{table}

\section{Conclusions and Discussions}\label{sec:conclusion}

We have developed a reduced dimension multiuser detection (RD-MUD) structure, assuming symbol-rate synchronization, which decreases the number of correlators at the front-end of a MUD receiver by exploiting the fact that the number of active users is typically much smaller than the total number of users in the system. \yx{The front-end of the RD-MUD} \ye{is motivated by analog CS and it} projects the received signal onto a lower dimensional detection subspace by correlating the received signal with a set of correlating signals. The correlating signals are constructed as linear combinations of the \yaor{signature} waveforms using a coefficient matrix $\A$, which determines the performance of RD-MUD and is our key design parameter. Based on the front-end output, RD-MUD detectors recover active users and their symbols in the detection subspace. 

We studied in detail two such detectors. The RDD detector, which is a linear detector that combines subspace projection along with thresholding for active user detection and RDDF detector, which is a nonlinear detector that combines decision-feedback matching pursuit for active user detection. We have shown that to achieve a desired probability-of-symbol-error, the number of correlators used by RD-MUD can be much smaller than that used by conventional MUD, and the complexity-per-bit of the RD-MUD detectors are not higher than their counterpart in the conventional MUD setting. In particular, when the random partial DFT matrix is used for the coefficient matrix $\A$ and the RDD and RDDF detectors are used for detection, the RD-MUD front-end requires a number of correlators proportional to log of the number of users, whereas the conventional MF-bank front-end requires a number of correlators equal to the number of users in the system.
We obtained theoretical performance guarantees for the RDD and RDDF detectors in terms of the coherence of $\A$, which are validated via numerical examples. 

\yx{In contrast to other work exploiting compressed sensing techniques for multiuser detection, our work has several distinctive features: (1) we consider analog \ag{received multiuser} signals; (2) we consider front-end complexity, which is the number of filters/correlators at the front-end to perform the analog-to-discrete conversion; (3) the noise is added in the analog domain \ag{prior to processing of the front-end}, so that the output noise vector can be colored due to front-end filtering; (4) \ag{we modify several \ye{conventional} compressed sensing estimation algorithms to make them applicable for symbol detection and study their probability-of-symbol-error performance}.}

Our results are based on binary modulation and can be extended to higher order modulation with symbols taking more possible values. In this case, however, the conditions to guarantee correct symbol detection may be stronger than the conditions to guarantee correct active user detection. We have also assumed that the signature waveforms are given. Better performance of RD-MUD might be obtained through joint optimization of the signature waveforms and the coefficient matrix $\A$. \yaor{Our results assume a synchronous channel model. Extending the ideas of this work to asynchronous channels perhaps using the methods developed in \cite{GedalyahuEldar2010} for time-delay recovery from low-rate samples, is a topic of future research. }

\section*{Acknowledgment}

The authors would like to thank Robert Calderbank and Lorne Applebaum for providing helpful suggestions with the numerical example regarding Kerdock codes.

\appendices

\section{Derivation of RD-MUD MMSE}\label{app:RD_MMSE}

Given the active user index set $\hat{\mathcal{I}}$ obtained from (\ref{support}), we define $\W = \A_{\hat{\mathcal{I}}}\R_{\hat{\mathcal{I}}}^2\A_{\hat{\mathcal{I}}}^H + \sigma^2\A\G^{-1}\A^H$, and $\bar{\M} = \R_{\hat{\mathcal{I}}}\A_{\hat{\mathcal{I}}}^H\W^{-1}$. We want to show that $\bar{\M}=\arg\min_{\M}\mathbb{E}\{\|\bb_{\hat{\mathcal{I}}} - \M\yb\|^2\}$. Using the same method for deriving the conventional MMSE detector of the MF-bank \cite{verduMUD1998}, 
we assume that $\bb_{\hat{\mathcal{I}}}$ has a distribution that is uncorrelated with the noise $\wb$ and that $\mathbb{E}\{\bb_{\hat{\mathcal{I}}} \bb_{\hat{\mathcal{I}}}^H\} = \I$. Based on $\hat{\mathcal{I}}$, we refer to the model (\ref{restrict}), and write the MSE as $\mathbb{E}\{\|\bb_{\hat{\mathcal{I}}} - \M\yb\|^2\} = \mbox{tr}(\mathbb{E}\{(\bb_{\hat{\mathcal{I}}}-\M\y)(\bb_{\hat{\mathcal{I}}}-\M\y)^H\})$. Now we expand
\begin{equation}
\begin{split}
&\quad \mathbb{E}\{(\bb_{\hat{\mathcal{I}}} -\M\y)(\bb_{\hat{\mathcal{I}}} -\M\y)^H\}\\
&=\mathbb{E}\{\bb_{\hat{\mathcal{I}}} \bb_{\hat{\mathcal{I}}} ^H\} -\mathbb{E}\{\bb_{\hat{\mathcal{I}}} \yb^H\}\M^H - \M \mathbb{E}\{\yb\bb_{\hat{\mathcal{I}}} ^H\}\\
&\quad+\M\mathbb{E}\{\yb\yb^H\}\M^H\\
&= \I + \M (\A_{\hat{\mathcal{I}}} \R_{\hat{\mathcal{I}}}^2 \A_{\hat{\mathcal{I}}}^H +\sigma^2\A\G^{-1}\A^H)\M^H \\
&\qquad-\R_{\hat{\mathcal{I}}}\A_{\hat{\mathcal{I}}}^H\M^H - \M\A_{\hat{\mathcal{I}}}\R_{\hat{\mathcal{I}}}.\label{corr_app}
\end{split}
\end{equation}
 It can be verified that $\M\A_{\hat{\mathcal{I}}}\R_{\hat{\mathcal{I}}} = \M \W\bar{\M}^H$. Hence from (\ref{corr_app}), we have
\begin{equation}
\begin{split}
 &\mathbb{E}\{(\bb_{\hat{\mathcal{I}}} -\M\y)(\bb_{\hat{\mathcal{I}}} -\M\y)^H\} \\
 &=
 \I + \M\W\M^H  - \bar{\M}\W\M^H - \M\W\bar{\M}^H\\
 & = \I - \bar{\M}\W\bar{\M}^H 
+ (\M-\bar{\M})\W(\M-\bar{\M})^H\\
 & = \I - \R_{\hat{\mathcal{I}}}\A_{\hat{\mathcal{I}}}^H
\W^{-1}\A_{\hat{\mathcal{I}}}\R_{\hat{\mathcal{I}}} \\
&\quad+ (\M-\bar{\M})\W(\M-\bar{\M})^H.
\end{split}
\label{appeqn1}
\end{equation}
Since $\W$ is a positive semidefinite matrix, the trace of the second term in (\ref{appeqn1}) is always nonnegative. Therefore, the matrix $\M$ that minimizes the MSE is $\bar{\M}$.

\section{Proof of Theorem \ref{thm_noisy}}
\label{proof_thm_noisy}

The proof of Theorem \ref{thm_noisy} for both the RDD and RDDF detectors are closely related. We therefore begin by proving several lemmas that are useful for both results. 

First, we prove that the random event $\mathcal{G}$ defined in (\ref{event})
occurs with high probability,
where $\tau$ is defined in (\ref{def_tau}). Then we show that when $\mathcal{G}$ occurs, both algorithms can detect the active users and their symbols. The proofs follow the arguments in \cite{Ben-HaimEldarElad2010} with modifications to account for the fact that $\wb$ is colored noise, and the error can also be caused by incorrect symbol detection. However, as we will show, the error probability of active user detection dominates the latter case.

\begin{lemma} \label{lemma_noise_bound}
Suppose that $\wb$ is a Gaussian random vector with zero mean and covariance $\sigma^2 \A\G^{-1}\A^H$. If $ N^{-(1+\alpha)}[\pi(1+\alpha)\log N]^{-1/2} \leq 1$ for some $\alpha > 0$, then the event $\mathcal{G}$ of (\ref{event}) occurs with probability  at least one minus (\ref{high_prob_noisy}).
\end{lemma}
\begin{proof}
The random variables $\{\ab_n^H \wb\}_{n=1}^N$ are jointly Gaussian, with means equal to zero, variances  $\sigma_n^2$ equal to $\sigma^2\ab^H_n \A\G^{-1}\A^H\ab_n$. 
Define
\begin{equation}
\hat{\tau} \triangleq \sigma[2(1+\alpha)\log N]^{1/2} \cdot\left[\max_n (\ab^H_n \A\G^{-1}\A^H\ab_n)\right]^{1/2}, \label{def_tau_hat}
\end{equation}
and an event
$
\hat{\mathcal{G}} \triangleq  \left\{\max_{1\leq n \leq N} |\ab_n^H \wb| < \hat{\tau}\right\}.$
Using Sidak's lemma  \cite{Sidak1967} , we have
\begin{equation}
\begin{split}
\mathbb{P}\left\{\hat{\mathcal{G}}\right\} &= \mathbb{P}\left\{|\ab_1^H \wb|< \hat{\tau}, \cdots, |\ab_N^H \wb|< \hat{\tau}\right\}\\
&\geq \prod_{n=1}^N \mathbb{P}\{|\ab_n^H \wb|< \hat{\tau}\}.
\end{split}
\label{hat_G}
\end{equation}

Since $\ab_n^H \wb$ is a Gaussian random variable with zero mean and variance $\sigma_n^2$, the tail probability of the colored noise can be written as
\begin{equation}
\mathbb{P}\{|\ab_n^H \wb| < \hat{\tau}) = 1 - 2Q\left(\frac{\hat{\tau}}{\sigma_n} \right\}.\label{q_bound}
\end{equation}
Using the bound on $Q(x)$: $Q(x)\leq (x\sqrt{2\pi})^{-1}e^{-x^2/2}$, (\ref{q_bound}) can be bounded as
$ \mathbb{P}\{|\ab_n^H \wb| < \hat{\tau}\} \geq 1 - \eta_n,$ where $\eta_n \triangleq \sqrt{2/\pi} (\sigma_n/\hat{\tau})e^{-\hat{\tau}^2/(2\sigma_n^2)}$.  
Define
$\sigma_{\max} \triangleq  \max_n \sigma_n = \sigma\left[\max_n (\ab_n^H\A\G^{-1}\A^H \ab_n)\right]^{1/2}$, 
$\eta_{\max} \triangleq \sqrt{2/\pi}(\sigma_{\max}/\hat{\tau})
e^{-\hat{\tau}^2/(2\sigma^2_{\max})}$.
Since ${\sigma_{\max}}/{\hat{\tau}} = [2(1+\alpha)\log N]^{-1/2}$ by the definition of $\hat{\tau}$,
we have $\eta_{\max} = \sqrt{2/\pi}[2(1+\alpha)\log N]^{-1/2}e^{-(1+\alpha)\log N}$.
It is easy to show that $\eta_n$ increases as $\sigma_n$ increases. Hence $\eta_n \leq \eta_{\max}$.
%
When $\eta_{\max} \leq 1$, we can use the inequality $(1-x)^N \geq 1- Nx$ when $x\geq 0$ and substitute the value of $\eta_{\max}$ to write (\ref{hat_G}) as
\begin{equation}
\begin{split}
&\mathbb{P}\{\hat{\mathcal{G}}\} \geq \prod_{n=1}^N (1-\eta_n) \geq (1-\eta_{\max})^N
\geq 1 - N\eta_{\max} \\
&\qquad= 1 - N^{-\alpha}[\pi(1+\alpha)\log N]^{-1/2},
\end{split}
   \label{G_bound}
\end{equation}
which holds for any $\eta_{\max} \leq 1$
and $N \geq 1$.
Next we show that $\hat{\tau}\leq \tau$. Note that
\begin{equation}
\begin{split}
&\ab_n^H \A\G^{-1}\A^H\ab_n  \\
&
\leq \|\A^H\ab_n\|^2 \lambda_{\max}(\G^{-1})  \\
& \leq [\max_n \left(\ab_n^H\A\A^H\ab_n\right)]\lambda_{\max}(\G^{-1}).
\end{split}
\label{34}
\end{equation}
From inequality (\ref{34}) and definitions (\ref{def_tau}) for $\tau$ and (\ref{def_tau_hat}) for $\hat{\tau}$, we obtain $\hat{\tau}\leq \tau$. Hence 
\begin{equation}
\begin{split}
&\mathbb{P}\{\mathcal{G}\} = \mathbb{P}\{\max_n |\ab_n^H \wb| < \tau\} \\
&\geq \mathbb{P}\{\max_n |\ab_n^H \wb|< \hat{\tau}\} = \mathbb{P}\{\hat{\mathcal{G}}\}.
\end{split}
\label{eqn2}
\end{equation}
Combining (\ref{G_bound}) and (\ref{eqn2}), we conclude that $P(\mathcal{G})$ is greater than one minus the expression (\ref{high_prob_noisy}), as required.
\end{proof}

The next lemma shows that, under appropriate conditions, ranking the inner products between $\ab_n$ and $\yb$ is an effective method of detecting the set of active users. The proof of this lemma is adapted from Lemma 3 in \cite{Ben-HaimEldarElad2010} to account for the fact that the signal vector $\yb$ here can be complex as $\A$ can be complex. Since the real part contains all the useful information, to prove this lemma, we follow the proof for Lemma 3 in \cite{Ben-HaimEldarElad2010} while using the following inequality whenever needed: $|\Re[\ab_n^H\ab_m]| \leq |\ab_n^H \ab_m| \leq \mu$ for $n\neq m$, and $|\Re[\ab_n^H \wb]| \leq |\ab_n^H \wb|$. \yx{The proofs are omitted due to space limitations. Details of the proof can be found in \cite{Xie2011PhD}.}
\begin{lemma}\label{lemma_active_user}
Let $\bb$ be a vector with support $\mathcal{I}$ which consists of $K$ active users, and let $\yb = \A\R\bb + \w$ for a Gaussian noise vector $\wb$ with zero mean and covariance $\A\G^{-1}\A^H$. Define $|r_{\max}|$ and $|r_{\min}|$ as in (\ref{gain_def}), and suppose that
\begin{equation}
|r_{\min}| - (2K-1)\mu |r_{\max}| \geq 2\tau. \label{cond_1}
\end{equation}
Then, if the event $\mathcal{G}$ of (\ref{event}) occurs, we have 
$\min_{n\in\mathcal{I}} |\Re[\ab_n^H \yb]| > \max_{n\notin\mathcal{I}} |\Re[\ab_n^H \yb]|.$
If, rather than (\ref{cond_1}), a weaker condition holds:
\begin{equation}
|r_{\max}| - (2K-1)\mu |r_{\max}| \geq 2\tau.\label{cond_2}
\end{equation}
Then, if the event $\mathcal{G}$ of (\ref{event}) occurs, we have
$\max_{n\in\mathcal{I}} |\Re[\ab_n^H \yb]| > \max_{n\notin\mathcal{I}} |\Re[\ab_n^H \yb]|.$ 
\end{lemma}

The following lemma demonstrates that the sign detector can effectively detect transmitted symbols for the RDD and RDDF detectors. \yx{This Lemma bounds the second term in $P_e$ that has not been considered in \cite{Ben-HaimEldarElad2010}.}
\begin{lemma}\label{lemma_symbols}
Let $\bb$ be a vector with $b_n \in \{1, -1\}$, for $n \in \mathcal{I}$ and $b_n = 0$ otherwise, and let $\yb = \A\R\bb + \wb$ for a Gaussian noise vector $\wb$ with zero mean and covariance $\sigma^2 \A\G^{-1}\A^H$. Suppose that
\begin{equation}
|r_{\min}| - (K-1)\mu |r_{\max}| \geq \tau. \label{sign_cond_1}
\end{equation}
Then, if the event $\mathcal{G}$ occurs, we have
\begin{equation}
\sign(r_n \Re[\ab_n^H\yb]) = b_n, \qquad n\in \mathcal{I}.\label{sign_1}
\end{equation}
If, instead of (\ref{sign_cond_1}), a weaker condition
\begin{equation}
|r_{\max}| + |r_{\min}| - 2(K-1)\mu |r_{\max}| \geq 2\tau \label{sign_cond_2}
\end{equation}
holds, then under the event $\mathcal{G}$, we have
$\sign(r_{n_1} \Re[\ab_{n_1}^H\yb]) = b_{n_1},$ 
for 
\begin{equation}
n_1 = \arg\max_n |\Re[\ab_n^H \yb]|. \label{n_1}
\end{equation}
\end{lemma}
\begin{proof}
To detect correctly, for $b_n = 1$, $\Re[r_n \ab_n^H \yb]$ has to be positive, and for $b_n = -1$, $\Re[r_n \ab_n^H \yb]$ has to be negative. 

First assume $b_n = 1$. We expand $\Re[r_n \ab_n^H \yb]$, find the lower-bound and the condition such that the lower bound is positive. Substituting in the expression for $\yb$, using the inequality that $x + y + z \geq x - |y| - |z|$, 
under the event $\mathcal{G}$, we obtain
\begin{equation}
\begin{split}
    &\Re[r_n \ab_n^H \yb] \\
    &= |r_n|^2   + \sum_{m\neq n} b_m r_n r_m\Re\left[\ab_n^H \ab_m\right] + r_n \Re\left[ \ab_n^H \wb\right]\\
    &\geq |r_n| |r_{\min}|    -  \sum_{m\neq n} |r_n|| r_m||\Re\left[\ab_n^H \ab_m\right]|\\
    &\qquad-  |r_n||\Re\left[\ab_n^H \wb\right]|\\
    & >   |r_n| \left[|r_{\min}|   - (K-1)\mu  |r_{\max}| - \tau\right].       \end{split}\label{69}
\end{equation}
From (\ref{69}), $\Re[r_n \ab_n^H \yb] >0$ for $n\in\mathcal{I}$ if (\ref{sign_cond_1}) holds and $b_n = 1$.

Similarly, we can show for $b_n = -1$, under event $\mathcal{G}$, if (\ref{sign_cond_1}) holds, $\Re[r_n \ab_n^H \yb] < 0$. Hence if (\ref{sign_cond_1}) holds we obtain (\ref{sign_1}).

Recall that $n_0$ is the index of the largest gain: $|r_{n_0}|= |r_{\max}|$. Due to (\ref{n_1}), we have
\begin{equation}
|\Re[\ab_{n_1}^H\yb]| \geq |\Re[\ab_{n_0}^H\yb]|. \label{key_OMP}
\end{equation}
We will show that under the event $\mathcal{G}$, once (\ref{sign_cond_2}) holds, then $\sign(r_{n_1} \Re[\ab_{n_1}^H\yb])\neq b_{n_1}$ leads to a contradiction to (\ref{key_OMP}).
First assume $b_{n_1} = 1$. If $\hat{b}_{n_1} = \sign(r_{n_1}\Re[\ab_{n_1}^H \yb])  \neq b_{n_1}$, then 
\begin{equation}
\begin{split}
    &\hat{b}_{n_1} \\
 &   = \sgn\left(r_{n_1}^2  + \sum_{m\neq n_1} b_m r_{n_1} r_m \Re\left[\ab_{n_1}^H \ab_m\right] + r_{n_1}\Re\left[ \ab_{n_1}^H \wb\right]\right) \\&= -1. \end{split}\label{det}
\end{equation}
So the expression inside the $\sgn$ operator of (\ref{det})  must be negative. Since $r_{n_1}^2 > 0$, we must have
\begin{equation}
\sum_{m\neq n_1} b_m r_{n_1} r_m \Re\left[\ab_{n_1}^H \ab_m\right] + r_{n_1}\Re\left[ \ab_{n_1}^H \wb\right] < 0. \label{inter}
\end{equation}

Multiplying the left-hand-side of (\ref{key_OMP}) by $|r_{n_1}|$, and using the equality $|x|\cdot|y| = |xy|$, we obtain
\begin{equation}
\begin{split}
&|r_{n_1}| |\Re[\ab_{n_1}^H \yb]|\\
&= |r_{n_1}| \left|r_{n_1} + \sum_{m\neq n_1} b_m r_m \Re[\ab_{n_1}^H \ab_m] + \Re[\ab_{n_1}^H \wb]\right|\\
&= \left|r_{n_1}^2 + \sum_{m\neq n_1} b_m r_{n_1} r_m \Re[\ab_{n_1}^H \ab_m] + r_{n_1}\Re[\ab_{n_1}^H \wb]\right|. 
\end{split}
\label{82}
\end{equation} 

Due to (\ref{det}), the last line of (\ref{82}) inside the $|\cdot|$ operator is negative. Using the fact that $r_{n_1}^2 > 0$ and (\ref{inter}), and the identity $|x+y| = -(x + y) = |y| -x$ when $x+y < 0$ and $y<0$, under the event $\mathcal{G}$, we obtain that 
\begin{equation}
\begin{split}
&|r_{n_1}| |\Re[\ab_{n_1}^H \yb]| \\
&= \left|\sum_{m \neq n_1} b_m r_{n_1} r_m \Re\left[\ab_{n_1}^H \ab_m\right]
+ r_{n_1}\Re\left[\ab_{n_1}^H \wb\right]\right|- r_{n_1}^2\\
&< |r_{n_1}|(K-1)\mu |r_{\max}| + |r_{n_1}|\tau - |r_{n_1}||r_{\min}|\\
&=   |r_{n_1}|[(K-1)\mu |r_{\max}| + \tau - |r_{\min}|]. 
\end{split}
\label{eqn84}
\end{equation}

On the other hand, multiply the right-hand-side of (\ref{key_OMP}) by $|r_{n_1}|$. Similarly, using the equality $|x|\cdot |y| = |xy|$ and triangle inequality, under the event $\mathcal{G}$, we obtain
\begin{equation}
\begin{split}
&|r_{n_1}||\Re[\ab_{n_0}^H \yb]|\\
&= \left|r_{n_1} r_{n_0} b_{n_0} + \sum_{m \neq n_0} b_m r_{n_1} r_m \Re\left[\ab_{n_0}^H \ab_m\right] + r_{n_1}\Re\left[\ab_{n_0}^H \wb\right]\right|\\
&> 
|r_{n_1}|[|r_{\max}| -  (K-1)\mu |r_{\max}|  - \tau]. 
\end{split}
\label{eqn86}
\end{equation}

Combining (\ref{eqn84}) and (\ref{eqn86}), we have that once (\ref{sign_cond_2}) holds, if $b_{n_1}=1$, then $\sgn(r_{n_1} \Re[\ab_{n_1}^H \yb]) = -1$ leads to $|\Re[\ab_{n_1}^H \yb]| <|\Re[\ab_{n_0}^H \yb]|$,
which contradicts (\ref{key_OMP}), and hence $\sgn(r_{n_1} \Re[\ab_{n_1}^H \yb]) = 1$. A similar argument can be made for $b_{n_1} = -1$, which completes the proof.

\end{proof}

We are now ready to prove Theorem \ref{thm_noisy}. The proof for the RDD detector is obtained by combining Lemmas \ref{lemma_noise_bound}, \ref{lemma_active_user} and \ref{lemma_symbols}. Lemma \ref{lemma_noise_bound} ensures that the event $\mathcal{G}$ occurs with probability at least as high as one minus (\ref{high_prob_noisy}). Whenever $\mathcal{G}$ occurs, Lemma \ref{lemma_active_user} guarantees by using (\ref{support}), that the RDD detector can correctly detect active users under the condition (\ref{cond_thresholding}), i.e. $\mathcal{G}\subset \{\hat{\mathcal{I}} = \mathcal{I}\}$. 
Finally, whenever $\mathcal{G}$ occurs, Lemma \ref{lemma_symbols} guarantees that, based on the correct support of active users, their transmitted symbols can be detected correctly under the condition (\ref{sign_cond_1}), i.e. $\mathcal{G}\subset \{\hat{b}_n = b_n, n\in\mathcal{I}\}$. Clearly condition (\ref{sign_cond_1}) is weaker than (\ref{cond_thresholding}), since (\ref{cond_thresholding}) can be written as $|r_{\min}|-(K-1)\mu |r_{\max}| \geq \tau + (\tau + K\mu |r_{\max}|) > \tau$, and hence if (\ref{cond_thresholding}) holds then (\ref{sign_cond_1}) also holds. In summary, under condition (\ref{cond_thresholding}), $\mathcal{G}\subset \{\hat{\mathcal{I}}=\mathcal{I}\}\cap \{\hat{\bb} = \bb\}$, and $1-P_e  = P(\{\hat{\mathcal{I}}=\mathcal{I}\}\cap \{\hat{\bb} = \bb\}) \geq P(\mathcal{G})$, which is greater than one minus (\ref{high_prob_noisy}), which concludes the proof for the RDD detector.

The proof for RDDt is similar to that for RDD detector and inspired by the proof of Theorem 1 in \cite{BajwaCalderbankJafarpour2010}.  Using similar arguments to Lemma \ref{lemma_active_user}, we can demonstrate that, when the number of active users $K \leq K_0$, when $\mathcal{G}$ occurs,
\begin{equation}
\begin{split}
&\min_{n\in\mathcal{I}}|\mathbb{R}[\ab_n^H \yb]| 
> |r_{\min}| - (K-1)\mu|r_{\max}| - \tau \\&\qquad\geq |r_{\min}| - (K_0-1)\mu|r_{\max}| - \tau,
\end{split}
\end{equation}
and
\begin{equation}
\max_{n\notin\mathcal{I}}|\mathbb{R}[\ab_n^H \yb]| < K\mu|r_{\max}| + \tau \leq K_0\mu|r_{\max}| + \tau.
\end{equation}
If (\ref{cond_1}) holds for $K = K_0$, we can choose a threshold $\xi$ such that $K_0\mu|r_{\max}| + \tau < \xi < |r_{\min}| - (K_0-1)\mu|r_{\max}| - \tau$. Then $\min_{n\in\mathcal{I}}|\mathbb{R}[\ab_n^H \yb]| > \xi$ and $\max_{n\notin\mathcal{I}}|\mathbb{R}[\ab_n^H \yb]| < \xi$, and hence for such $\xi$ the RDDt detector can correctly detect the active users with high probability. Since when (\ref{cond_1}) holds, (\ref{sign_cond_1}) is true, from Lemma \ref{lemma_symbols} we know the symbol can be correctly detected with high probability as well. 

We now prove the performance guarantee for the RDDF detector adopting the technique used in proving Theorem 4 in \cite{Ben-HaimEldarElad2010}. First we show that whenever $\mathcal{G}$ occurs, the RDDF detector correctly detects an active user in the first iteration, which follows from Lemmas \ref{lemma_noise_bound} and \ref{lemma_active_user}. Note that (\ref{cond_OMP}) implies (\ref{cond_2}), and therefore, by Lemma \ref{lemma_active_user}, we have that by choosing the largest $|\Re[\ab_n^H \yb]|$, the RDDF detector can detect a correct user in the set $\mathcal{I}$.  Second, we show that whenever $\mathcal{G}$ occurs, the RDDF detector correctly detects the transmitted symbol of this active user. Note that (\ref{cond_OMP})  also implies (\ref{sign_cond_2}), since (\ref{cond_OMP}) can be written as $|r_{\min}|\geq 2\tau/[1-(2K-1)\mu]$, which implies $|r_{\max}|\geq 2\tau/[1-(2K-1)\mu]$, and hence $|r_{\max}| + |r_{\min}| - 2(K-1)\mu |r_{\max}| \geq 2\tau [1-2(K-1)\mu]/[1-(2K-1)\mu] + |r_{\min}| > 2\tau$, since $[1-2(K-1)\mu]/[1-(2K-1)\mu] \geq1$. Therefore, by Lemma \ref{lemma_symbols}, using a sign detector, we can detect the symbol correctly. Consequently, the first step of the RDDF detector correctly detect the active user and its symbol, i.e. $\mathcal{G}\subset \{{\mathcal{I}}^{(1)} \subset \mathcal{I}, {b}_{n_1}^{(1)} = b_{n_1}\}$.

The proof now continues by induction. Suppose we are currently in the $k$th iteration of the RDDF detector, $1\leq k \leq K$, and assume that $k-1$ correct users and their symbols have been detected in all the $k-1$ previous steps. The $k$th step is to detect the user with the largest $|\Re[\ab_n^H \vb^{(k-1)}]|$. Using the same notations as those in Section \ref{sec:algorithm_II} and by definition of $\vb^{(k-1)}$, we have
\begin{equation} 
\vb^{(k-1)} = \A \R (\bb - \bb^{(k-1)}) + \wb = \A\R\xb^{(k-1)} + \wb, \label{model_2}
\end{equation}
where $\xb^{(k-1)}\triangleq \bb - \bb^{(k-1)}$. This vector has support $\mathcal{I}/\mathcal{I}^{(k-1)}$ and has at most $K-k+1$ non-zero elements, since $\bb^{(k-1)}$ contains correct symbols at the correct locations for $(k-1)$ active users, i.e. ${b}^{(k-1)}_n = b_n$, for $ n\in {{\mathcal{I}}^{(k-1)}}$. This $\vb^{(k-1)}$ is a noisy measurement of the vector $ \A \R\xb^{(k-1)}$. The data model in (\ref{model_2}) for the $k$th iteration is identical to the data model in the first iteration with $\bb$ replaced by $\xb^{(k-1)}$ (with a smaller sparsity $K-k+1$ rather than $K$), $\mathcal{I}$ replaced by $\mathcal{I}/{\mathcal{I}}^{(k-1)}$, and $\yb$ replaced by $\vb^{(k-1)}$.
Let $|r_{\max}^{(k)}|\triangleq \max_{n\in\mathcal{I}/{\mathcal{I}}^{(k-1)}} |r_n|$. By assumption, $k-1$ active users with largest gains have been correctly detected in the first $k-1$ rounds, and hence $|r_{\max}^{(k)}| = |r^{(k)}|$. Since
\begin{equation}
|r^{(k)}| \geq |r_{\min}|,\label{k_iteration}
\end{equation}
we have that under condition (\ref{cond_OMP}) this model (\ref{model_2}) also satisfies the requirement (\ref{cond_2}). Consequently, by Lemma \ref{lemma_active_user}, we have that under the event $\mathcal{G}$,
$\max_{n \in \mathcal{I}/{\mathcal{I}}^{(k-1)}} |\Re[\ab_n^H  \vb^{(k-1)}]| > \max_{n \in (\mathcal{I}/{\mathcal{I}}^{(k-1)})^c} |\Re[\ab_n^H \vb^{(k-1)}]|$.
Therefore, in the $k$th iteration, the RDDF detector can detect an active user correctly, i.e. $\mathcal{G}\subset \{{\mathcal{I}}^{(k)} \subset \mathcal{I}\}$, and no index of active users that has been detected before will be chosen again. On the other hand, since (\ref{cond_OMP}) can be written as $|r_{\min}| \geq 2\tau/[1-(2K-1)\mu]$, from (\ref{k_iteration}) this implies $|r^{(k)}| \geq 2\tau/[1-(2K-1)\mu]$, and hence $|r^{(k)}| - (2K-1)\mu |r^{(k)}|\geq 2\tau$, and consequently $|r^{(k)}| - (2K-2)\mu |r^{(k)}| + |r_{\min}|\geq 2\tau$. Consequently, condition (\ref{sign_cond_2}) is true for (\ref{model_2}). Then by Lemma \ref{lemma_symbols}, we have that under the event $\mathcal{G}$,
$\sgn(r_{n_k}\Re[\ab_{n_k}^H \vb^{(k-1)}]) = b_{n_k}$,
i.e. $\mathcal{G}\subset \{{b}_{n_k}^{(k)} = b_{n_k}\}$.
By induction, since no active users will be detected twice, it follows that the first $K$ steps of the RDDF detector can detect all active users and their symbols, i.e. 
\begin{equation}
\begin{split}
&\mathcal{G}\subset \cup_{k=1}^K \{{\mathcal{I}}^{(k)} \subset \mathcal{I},{b}_{n_k}^{(k)} = b_{n_k}\} \\
&= \{{\mathcal{I}}^{(K)}=\mathcal{I}, {b}_{n}^{(K)} = b_{n}, n\in\mathcal{I}^{(K)}\}. 
\end{split}\label{eqn_99}
\end{equation} 
Note that condition (\ref{cond_2}) is weaker than (\ref{cond_OMP}), since (\ref{cond_OMP}) can be written as $|r_{\min}|[1-(2K-1)\mu] \geq 2\tau$, which implies $|r_{\max}|[1-(2K-1)\mu] \geq 2\tau$. This further implies $|r_{\max}|[1-2(K-1)\mu] + |r_{\min}| \geq 2\tau$, since $1-2(K-1)\mu\geq 1-(2K-1)\mu$ and $|r_{\min}|\geq 0$.
Consequently, under condition (\ref{cond_OMP}), from (\ref{eqn_99}), $\mathcal{G}\subset \{\hat{\mathcal{I}}=\mathcal{I}\}\cap \{\hat{\bb} = \bb\}$, and $1-P_e  = \mathbb{P}\{\{\hat{\mathcal{I}}=\mathcal{I}\}\cap \{\hat{\bb} = \bb\}\} \geq \mathbb{P}\{\mathcal{G}\}$ which is greater than one minus (\ref{high_prob_noisy}), which concludes the proof for the RDDF detector. 

The proof for RDDFt  follows the above proof for RDDF with one more step. Note that when we have correctly detected all active users in $K$ rounds, from (\ref{model_2}) the residual $\vb^{K} = \wb$ contains only noise. Hence, when $\mathcal{G}$ occurs, $\|\A^H\vb^{K}\|_\infty = \|\A^H \wb\|_\infty = \max_{1\leq n \leq N} |\ab_n^H \wb| < \tau$, from Lemma \ref{lemma_noise_bound}. On the other hand, in the $k$-th round, $k = 1, \ldots, K$, from (\ref{model_2}), we have that when $\mathcal{G}$ occurs
\begin{eqnarray}
&&\|\A^H\vb^{(k-1)}\|_\infty\nonumber \\
&=& \max_{1\leq n \leq N} \left|\sum_{m\in \mathcal{I}/\mathcal{I}^{(k-1)}} r_m b_m \ab_n^H \ab_m + \ab_n^H \wb\right|\\
&>& |r_{\max}^{(k)}| - (K-k)\mu|r_{\max}^{(k)}| - \tau > 0. \label{lower1}
\end{eqnarray}
The expression in (\ref{lower1}) is positive, when (\ref{cond_OMP}) holds (recall that (\ref{cond_OMP}) is also required to detect correct active users): because when $|r_{\min}| - (2K-1)\mu|r_{\min}| > 2\tau$, since $|r^{(k)}| \geq |r_{\min}|$, $|r^{(k)}| - (2K-1)\mu|r^{(k)}| > 2\tau$, and hence $|r^{(k)}| - (2K-2k-1)\mu|r^{(k)}| > 2\tau \geq 0$. Therefore when (\ref{cond_OMP}) holds, we can choose $\epsilon < \min_{k=1}^K \{|r^{(k)}|[1 - (K-k)\mu] - \tau\} < r_{\min} - \tau$.  
Therefore, under the condition (\ref{cond_OMP}), when $\mathcal{G}$ occurs, we can choose $\tau <\epsilon < \min_{k=1}^K \{|r^{(k)}|[1 - (K-k)\mu]| - \tau\}$, so that $\|\A^H\vb^{(k-1)}\|_\infty > \epsilon$, $k = 1, \ldots, K$, and $\|\A^H\vb^{(K)}\|_\infty < \epsilon$. Finally, using similar arguments as for RDDF that (\ref{cond_OMP}) guarantees (\ref{sign_cond_2}), RDDFt can also correctly detect the symbols with high probability.

This completes the proof of Theorem \ref{thm_noisy}.

\bibliography{yao}

\end{document}